\documentclass[11pt]{article}  
\usepackage{amsmath}
\usepackage{amssymb}
\usepackage{amsthm}
\usepackage{natbib}
\usepackage{bm,upgreek}
\usepackage{graphicx}
\usepackage{epstopdf}
\usepackage{booktabs}
\usepackage{rotating}
\usepackage{threeparttable}
\usepackage[margin=1.5in, marginparwidth=2cm]{geometry}
\usepackage{authblk}
\usepackage{setspace}
\usepackage{floatrow}
\usepackage{footmisc}
\usepackage{multirow}
\usepackage{makecell}
\usepackage{tabu}
\usepackage{longtable}
\usepackage{threeparttablex}
\usepackage{soul}
\usepackage{algorithm}
\usepackage{algpseudocode}
\usepackage{algorithmicx}
\usepackage{enumitem}
\usepackage{subfig}
\usepackage[draft]{todonotes}

\newcommand{\var}{\text{Var}}

\theoremstyle{plain}

\newtheorem*{proof*}{Proof}
\newtheorem{corollary}{Corollary}

\title{Inference in experiments conditional on observed imbalances in covariates\thanks{We are grateful for comments and suggestions from seminar participants at IFAU and Stockholm University.}}
\author{Per Johansson\thanks{Uppsala University, IFAU, Tsinghua University} and Mattias Nordin\thanks{Uppsala University, UCFS}}
\begin{document}
\onehalfspace

\maketitle

\begin{abstract}
\noindent Double blind randomized controlled trials are traditionally seen as the gold standard for causal inferences as the difference-in-means estimator is an unbiased estimator of the average treatment effect in the experiment. The fact that this estimator is unbiased over all possible randomizations does not, however, mean that any given estimate is close to the true treatment effect. Similarly, while pre-determined covariates will be balanced between treatment and control groups on average, large imbalances may be observed in a given experiment and the researcher may therefore want to condition on such covariates using linear regression. This paper studies the theoretical properties of both the difference-in-means and OLS estimators \emph{conditional} on observed differences in covariates. By deriving the statistical properties of the conditional estimators, we can establish guidance for how to deal with covariate imbalances. We study both inference with OLS, as well as with a new version of Fisher's exact test, where the randomization distribution comes from a small subset of all possible assignment vectors.
\end{abstract}

\section{Introduction}

Double blind randomized controlled trials (RCT) are traditionally seen as the gold standard for causal inferences as it provides probabilistic inference of the unbiased difference-in-means estimator under no model assumption (cf. \cite{freedman_regression_2008}). This concept of unbiasedness of an estimator is however often misunderstood as the estimate being ``the truth'' \citep[cf.][]{Deaton&Cartwright:2018}. In a single experiment the estimate may still be very far from the true effect due to an, unfortunate, bad treatment assignment.

The reason for the unique position of the RCT in the research community is that it provides an objective and transparent strategy for conducting an empirical study, not necessarily that it is most efficient way of scientific learning.\footnote{The conflict between the use of RCT, propagated by Ronald Fisher, and other strategies was discussed early on, see e.g. \citet{Student:1938}. The paper was published, with the help of Egon Pearson and Jerzy Neyman, after the death of Gosset in 1937. For an interesting discussion on decision theory and the motivation for randomization see \citet{Banerjee&etal:2017}.} To facilitate the transparency, it is common practice in scientific journals that researchers present imbalances of pre-experimental covariates of the treated and controls, typically showing the means and standard deviations of these covariates. Of course, as pointed out by \cite{mutz_perils_2019}, if one knows that treatment is randomly assigned, there is no such thing as a ``failed'' randomization (in a completely randomized design, any treatment assignment is possible) which means that any large imbalance in observed covariates does not necessitate any further action.

Indeed, \cite{mutz_perils_2019} argue that by studying balance on observed covariates, researchers run the risk of making their results \emph{less} credible as researchers may be temped to adjust for observed imbalances, which compromises the inference. By doing so, they may also estimate several different models, raising the concern of ``p-hacking''. At the same time, removing descriptive tables of balances between treated and controls does not seem to be possible given that the transparency of the research design is an important reason for using an RCT.

Furthermore, while it is true that the difference-in-means estimator is an unbiased estimator over all possible randomizations, this fact may be of little solace to the applied researcher who have conducted an experiment in which he/she have observed imbalances, as imbalances may indicate that the estimate is far from the true value. 

In this paper, we provide a framework for conditional inference that are not compromised by conditioning on covariates. We derive the distributions of different treatment effect estimators conditional on covariate imbalances to establish guidance for how to deal any observed imbalances. We also discuss how to perform both Neyman-Pearson and Fisher tests which give correct inferences regardless of whether covariate imbalances have been observed. 

Different from \cite{mutz_perils_2019}, who considers inference to the population conditional on imbalance in a single covariate, we consider randomization inference to the sample conditional on observed imbalances in a vector of covariates. By focusing on randomization inference, i.e., that the stochasticity comes random treatment assignment rather than random sampling, we follow, among others, \cite{freedman_regression_2008}, \cite{cox_randomization_2009} and \cite{lin_agnostic_2013} who study unconditional inference to the sample.\footnote{See also \cite{miratrix_adjusting_2013}, who study conditional inference when using post-stratification.} We consider both homogeneous and heterogeneous treatment effects and show that when explanatory covariates are imbalanced, the difference-in means estimator is biased while the conditional OLS estimator is close to unbiased. The variance of the conditional OLS estimator is increasing with the imbalance of the covariates and in the number of covariates. Thus, in an experiment there is a trade-off between bias and variance reduction in how many covariates to adjust for, and the trade-off depends on the imbalance of the covariates as well as the importance of the covariates in explaining the outcome. 



In situations with a large set of covariates relative to the sample size we provide algorithms for covariate-adjustments that do not suffer from the pitfalls pointed out by \cite{mutz_perils_2019}, where the procedures make use of the principal components of the covariates. Based on the imbalance of these principal components, the number of components to adjust for is chosen such that randomized inference can be justified. The algorithm for the conditional Fisher randomization test samples a sufficient number of similar treatment assignments and perform a Fisher test within that set.

The paper proceeds by presenting the theoretical justification for conditional inference under homogeneous treatment effects in the next section with Section 3 illustrating these results. Section 4 discusses the problem with a large set of covariates in comparison to sample size and presents the different algorithms together with Monte Carlo simulation results. In Section 5, we study the case with heterogeneous treatment effects both theoretically and with Monte Carlo simulations. Section 6 concludes the paper.

\section{Theoretical framework}
Consider a RCT with $n$ units in the sample, indexed by $i$, with $n_{1}$ to
be assigned to treatment and $n_{0}$ to be assigned to control. Let $W_{i}=1$
or $W_{i}=0$ if unit $i$ is assigned treatment or control, respectively, and
define $\mathbf{W}=\left[\begin{array}{ccc}W_{1}&\ldots&W_{n}\end{array}\right]^{\prime }.$ The set $%
\mathcal{W}=\{\mathbf{W}^{1},\ldots,\mathbf{W}^{n_A}\}$ contains all possible assignment vectors and has cardinality $|\mathcal{W}|=\tbinom{n}{n_{1}}=n_A$.

Let $Y_{i}(w)$ denote the potential outcome for unit $i$ given the
treatment ($w=1)$ and control ($w=0)$. We assume no interference
between individuals and the same treatment (i.e. SUTVA)\ which means that
the observed outcome is $Y_{i}\equiv Y(W_{i}).$ The estimand of interest is the sample average treatment effect defined as
\begin{equation}
  	\tau = \frac{1}{n}\sum_{i=1}^n Y_i(1) - Y_i(0).
\end{equation} 
The difference-in-means estimator is
\begin{equation}
\widehat{\tau }_{DM}=\overline{Y}_{1}-\overline{Y}_{0},  \label{Estimator}
\end{equation}%
where 
\[
\overline{Y}_{w}=\frac{1}{n_{w}}\sum_{i:W_{i}=w}^{n_{w}}Y_{i},\quad w=0,1.
\]

Let $\mathbf{Z}$ be the $n\times K$ matrix of fixed covariates in the sample. Define the linear projection
in the sample 
\begin{equation}
Y_{i}(0)=\alpha + \mathbf{z}_{i}^{\prime }\bm{\upbeta }+\varepsilon _{i},
\label{Model}
\end{equation}%
where $\varepsilon _{i}$ is a fixed residual. Define $%
\overline{\mathbf{z}}_{w}=\frac{1}{n_{w}}\sum_{i:W_{i}=w}^{n_{w}}\mathbf{z}%
_{i}$ and $\overline{\mathbf{z}}=\frac{1}{n}\sum_{i=1}^{n}%
\mathbf{z}_{i}$. The difference-in-means estimator can be written as
\[
\widehat{\tau }_{DM}=\overline{Y}_{1}-\overline{Y}_{0}=\tau+(\overline{\mathbf{z}}%
_{1}-\overline{\mathbf{z}}_{0})^{\prime }\bm{\upbeta}+\overline{%
\varepsilon }_{1}-\overline{\varepsilon }_{0},
\]%
where $\overline{\varepsilon }_{w}=\frac{1}{n_{w}}\sum_{i:W=w}^{n_{w}}%
\varepsilon _{i}.$ As $W$ is random, both $\overline{\mathbf{z}%
}_{w}$ and $\overline{\varepsilon }_{w}$ are random even though $\mathbf{Z}$
and $\bm{\upvarepsilon}$ are fixed.

Let $E_{\mathcal{W}}(\cdot)$ and $V_{\mathcal{W}}(\cdot)$ denote expectation and variance over
randomizations in the set $\mathcal{W}$, and $\overline{Y(0)}=\frac{1}{n}\sum_{i=1}^n Y_i(0)$. It is the case that
\begin{equation}
	E_{\mathcal{W}}(\widehat{\tau }_{DM})=\tau,
\end{equation}
and
\begin{align}
V_{\mathcal{W}}(\widehat{\tau }_{DM}) &=\bm{\upbeta}^{\prime }V_{\mathcal{W}}(\overline{%
\mathbf{z}}_{1}-\overline{\mathbf{z}}_{0})^{\prime }\bm{\upbeta} +V_{\mathcal{W}}(%
\overline{\varepsilon }_{1}-\overline{\varepsilon }_{0})  \nonumber \\
&=\frac{n}{n_{0}n_{1}}\bm{\upbeta}^{\prime }%
\left(\widetilde{\mathbf{Z}}^{\prime }\widetilde{\mathbf{Z}}/(n-1)\right)\bm{\upbeta}+\frac{n}{%
n_{0}n_{1}} \frac{1}{n-1}\sum_{i=1}^n \varepsilon_i^2
 \nonumber \\
 & =\frac{n}{n_{0}n_{1}}\frac{1}{n-1}\sum_{i=1}^n (Y_i(0)-\overline{Y(0)})^2,
\label{VarDM}
\end{align}%
where $V_{\mathcal{W}}(\overline{\mathbf{z}}_{1}-\overline{\mathbf{z}}_{0})=V_{\mathcal{W}}(%
\mathbf{z})=\mathbf{\Sigma }_{z}$, $\widetilde{\mathbf{Z}}=\mathbf{Z}-\overline{\mathbf{z}}^{\prime }$ and 
$\widetilde{\mathbf{Y}}\mathbf{(0)}=\mathbf{Y(0)}-\overline{Y(0)},$ and%
\begin{equation}
\bm{\upbeta =}(\widetilde{\mathbf{Z}}^{\prime }\widetilde{\mathbf{Z}})^{-1}%
\widetilde{\mathbf{Z}}^{\prime }\widetilde{\mathbf{Y}}\mathbf{(0)}.
\end{equation}%
Note that $\widetilde{\mathbf{Z}}^{\prime }\widetilde{\mathbf{Z}}$ is
observed in data$.$ $\bm{\upbeta }$ is however not observed as $%
Y_{i}(0)$ is not observed if $W_{i}=1$.

We are interested in the stochastic properties of the difference-in-means estimator when $\overline{\mathbf{z}}_{1}\mathbf{-}\overline{\mathbf{z}}_{0}$ is held at some fixed value. Let $\mathcal{W}_{\bm{\Delta}} \subseteq \mathcal{W}$ be the set of assignments for which $\overline{\mathbf{z}}_{1}\mathbf{-}\overline{\mathbf{z}}_{0}=\bm{\Delta}$. $E_{\mathcal{W}_{\bm{\Delta}}}(\cdot)$ and $V_{\mathcal{W}_{\bm{\Delta}}}(\cdot)$ denote expectation and variance over randomizations in this set. We have
\begin{equation}
E_{\mathcal{W}_{\bm{\Delta}} }(\widehat{\tau }_{DM}) =\tau +\bm{\Delta}^{\prime }\bm\upbeta +E_{\mathcal{W}_{\bm{\Delta}} }(\overline{%
\varepsilon }_{1}-\overline{\varepsilon }_{0}), \label{eq:exp_dm}
\end{equation}%
and
\begin{equation}
	V_{\mathcal{W}_{\bm{\Delta}} }(\widehat{\tau }_{DM}) =V_{\mathcal{W}_{\bm{\Delta}} }(\overline{%
	\varepsilon }_{1}-\overline{\varepsilon }_{0}).\label{eq:var_cond_dm}
\end{equation}
Note that we cannot in general say that $E_{\mathcal{W}_{\bm{\Delta}} }(\overline{%
\varepsilon }_{1}-\overline{\varepsilon }_{0})=0$, and it is also the case that $V_{\mathcal{W}_{\bm{\Delta}} }(\overline{%
	\varepsilon }_{1}-\overline{\varepsilon }_{0})$ is not a constant, but depend on $\bm{\Delta}$. In the Appendix, we derive the explicit formula for $E_{\mathcal{W}_{\bm{\Delta}} }(\overline{%
\varepsilon }_{1}-\overline{\varepsilon }_{0})$ and $V_{\mathcal{W}_{\bm{\Delta}} }(\overline{\varepsilon }_{1}-\overline{\varepsilon }_{0})$ when $\mathbf{Z}$ consists of a single dummy variable. We there show that, in a balanced experiment, the variance is at its maximum when $\bm\Delta=0$ and decreases symmetrically as the magnitude of $\bm\Delta$ increases.


Let $\overline{Y}=\frac{1}{n}\sum_{i=1}^n Y_i$, it is helpful to define $\widetilde{\mathbf{Y}}=\mathbf{Y}-\overline{Y}$, $\widetilde{\mathbf{%
W}}=\mathbf{W}-n_1/n$ and $\widetilde{\mathbf{M}}_{z}=\mathbf{I}-\widetilde{%
\mathbf{Z}}\mathbf{(\widetilde{\mathbf{Z}}}^{\prime }\widetilde{\mathbf{Z}}%
\mathbf{)}^{-1}\widetilde{\mathbf{Z}}^{\prime }$. The OLS\
estimator of the treatment effect from regressing the outcome on a treatment indicator, controlling for the covariates $\mathbf{Z}$, is
\begin{equation}
\widehat{\tau }_{z}=\mathbf{(\widetilde{\mathbf{W}}^{\prime }\widetilde{%
\mathbf{M}}}_z\widetilde{\mathbf{W}}\mathbf{)}^{-1}\widetilde{\mathbf{W}}%
^{\prime }\widetilde{\mathbf{M}}_z \widetilde{\mathbf{Y}},
\end{equation}
where 
\begin{equation}
	\widetilde{\mathbf{Y}} = \widetilde{\mathbf{Z}}\bm\upbeta + \widetilde{\mathbf{W}} \tau + \bm\upvarepsilon.
\end{equation}
Because $\widetilde{\mathbf{W}}^{\prime }\widetilde{\mathbf{M}}_z \widetilde{\mathbf{Z}}\bm\upbeta = \mathbf{0}$, we have

\begin{equation}
	\widehat{\tau }_{z} = \tau + \mathbf{(\widetilde{\mathbf{W}}^{\prime }\widetilde{%
\mathbf{M}}}_z\widetilde{\mathbf{W}}\mathbf{)}^{-1}\widetilde{\mathbf{W}}%
^{\prime }\widetilde{\mathbf{M}}_z \mathbf{\bm\upvarepsilon}.
\end{equation}

The numerator of the OLS estimator can be written as
\begin{align}
	\widetilde{\mathbf{W}}^{\prime }\widetilde{\mathbf{M}}_z \bm\upvarepsilon = \widetilde{\mathbf{W}}^{\prime }\bm\upvarepsilon - \widetilde{\mathbf{W}}^{\prime }\widetilde{\mathbf{Z}}\mathbf{(\widetilde{\mathbf{Z}}}^{\prime }\widetilde{\mathbf{Z}}%
	\mathbf{)}^{-1}\widetilde{\mathbf{Z}}^{\prime }\bm\upvarepsilon=\widetilde{\mathbf{W}}^{\prime }\bm\upvarepsilon =\frac{n_0n_1}{n} (\overline{\upvarepsilon}_1 - \overline{\upvarepsilon}_0),
\end{align}
as $\widetilde{\mathbf{Z}}^{\prime }\bm\upvarepsilon=\mathbf{0}$. The denominator of the OLS estimator can be written as
\begin{align}
	\widetilde{\mathbf{W}}^{\prime }\widetilde{\mathbf{M}}_z\widetilde{\mathbf{W}} &= 
	\widetilde{\mathbf{W}}^{\prime }(\mathbf{I}-\widetilde{%
	\mathbf{Z}}\mathbf{(\widetilde{\mathbf{Z}}}^{\prime }\widetilde{\mathbf{Z}}%
	\mathbf{)}^{-1}\widetilde{\mathbf{Z}}^{\prime })\widetilde{\mathbf{W}} \nonumber \\
	&= \widetilde{\mathbf{W}}^{\prime }\widetilde{\mathbf{W}} - \widetilde{\mathbf{W}}^{\prime }\widetilde{\mathbf{Z}}\mathbf{(\widetilde{\mathbf{Z}}}^{\prime }\widetilde{\mathbf{Z}}%
	\mathbf{)}^{-1}\widetilde{\mathbf{Z}}^{\prime }\widetilde{\mathbf{W}} \nonumber \\
	&= \frac{n_0n_1}{n} \left( 1 - \frac{n_0n_1}{n}\bm\Delta'(\widetilde{\mathbf{Z}'}\mathbf{\widetilde{\mathbf{Z}}})^{-1}\bm\Delta \right). 
\end{align}

Let $M_{\bm{\Delta}}:=\frac{n_0n_1}{n}\bm{\Delta}'(\widetilde{\mathbf{Z}'}\mathbf{\widetilde{\mathbf{Z}}}/(n-1))^{-1}\bm{\Delta}$ be the Mahalanobis distance between treatment and control in $\mathbf{Z}$. We get
\begin{equation}
	\widehat{\tau }_{z}=\mathbf{(\widetilde{\mathbf{W}}^{\prime }\widetilde{%
\mathbf{M}}}_z\widetilde{\mathbf{W}}\mathbf{)}^{-1}\widetilde{\mathbf{W}}%
^{\prime }\widetilde{\mathbf{M}}_z \widetilde{\mathbf{Y}} = \tau +  
\frac{\overline{\upvarepsilon}_1 - \overline{\upvarepsilon}_0}{ 1 - M_{\bm{\Delta}}/(n-1)}.
\end{equation}

Over all assignment vectors in $\mathcal{W}$, it is the case that $E_{\mathcal{W} }\left( \overline{\upvarepsilon}_1 - \overline{\upvarepsilon}_0 \right)=0$, and so
\begin{equation}
	E_{\mathcal{W}}(\widehat{\tau }_{z})=\tau.
\end{equation}
I.e., the OLS estimator is an unbiased estimator over all assignment vectors when treatment effects are homogeneous. The conditional expectation becomes
\begin{equation}
	E_{\mathcal{W}_{\bm{\Delta}} }(\widehat{\tau }_z)=\tau +  
\frac{E_{\mathcal{W}_{\bm{\Delta}} }(\overline{\upvarepsilon}_1 - \overline{\upvarepsilon}_0)}{ 1 - M_{\bm{\Delta}}/(n-1)},\label{eq:cond_e_reg_exact}
\end{equation}
with the variance being
\begin{equation}
	V_{\mathcal{W}_{\bm{\Delta}} }(\widehat{\tau }_{z}) = \frac{V_{\mathcal{W}_{\bm{\Delta}} }(\overline{%
	\varepsilon }_{1}-\overline{\varepsilon }_{0})}{ \left(1 - M_{\bm{\Delta}}/(n-1)\right)^2}.\label{eq:var_cond_reg}
\end{equation}
From equations \eqref{eq:var_cond_dm} and \eqref{eq:var_cond_reg}, we get
\begin{equation}
	\frac{V_{\mathcal{W}_{\bm{\Delta}} }(\widehat{\tau }_{DM})}{V_{\mathcal{W}_{\bm{\Delta}} }(\widehat{\tau }_{z})} = \left(1 - M_{\bm{\Delta}}/(n-1)\right)^2.\label{eq:var_cond_diff}
\end{equation}
This equation implies that for $\bm\Delta = 0$, the variance of the two estimators are identical. As the Mahalanobis distance increases, the variance of the difference-in-means estimator gets relatively smaller compared to the variance of the OLS estimator.

The conditional MSE for the difference-in-means estimator is
\begin{equation}
	MSE_{\mathcal{W}_{\bm{\Delta}} }(\widehat{\tau }_{DM}) =\bm\upbeta'\bm{\Delta}\bm{\Delta}^{\prime }\bm\upbeta + E_{\mathcal{W}_{\bm{\Delta}} }\left((\overline{%
 		\varepsilon }_{1}-\overline{\varepsilon }_{0})^2\right) + \bm{\Delta}^{\prime }\bm\upbeta E_{\mathcal{W}_{\bm{\Delta}} }(\overline{%
 		\varepsilon }_{1}-\overline{\varepsilon }_{0}),
\end{equation}
whereas the conditional MSE for the OLS estimator is
\begin{equation}
 	MSE_{\mathcal{W}_{\bm{\Delta}} }(\widehat{\tau }_z) =
 	\frac{E_{\mathcal{W}_{\bm{\Delta}} }\left((\overline{%
 		\varepsilon }_{1}-\overline{\varepsilon }_{0})^2\right)}{ \left(1 - M_{\bm{\Delta}}/(n-1)\right)^2}.
\end{equation} Therefore, we have that the conditional MSE of the difference-in-means estimator is greater than the conditional MSE of the OLS estimator if
\begin{equation}
	\bm\upbeta'\bm{\Delta}\bm{\Delta}^{\prime }\bm\upbeta + \bm{\Delta}^{\prime }\bm\upbeta E_{\mathcal{W}_{\bm{\Delta}} }(\overline{%
 		\varepsilon }_{1}-\overline{\varepsilon }_{0}) > E_{\mathcal{W}_{\bm{\Delta}} }\left((\overline{\varepsilon }_{1}-\overline{\varepsilon }_{0})^2\right)\left(\frac{1}{\left( 1 - M_{\bm{\Delta}}/ (n-1)\right)^{2}}-1\right).\label{eq:comp_dm_ols}
\end{equation}
Note that, once again, for $\bm\Delta=\mathbf{0}$, the MSE for the difference-in-means and OLS estimators are identical. For $\bm\Delta\neq\mathbf{0}$, as the sample size increases, $M_{\bm{\Delta}}/ (n-1)=r^2$ (the R-squared from the regression of $W$ on $\mathbf{Z}$) will tend to zero and the OLS estimator will always be more efficient as long as the covariates are relevant ($\bm\upbeta \neq \mathbf{0})$.

At this point, it is helpful to compare the expression in equation \eqref{eq:comp_dm_ols} with Theorem 1 in \cite{mutz_perils_2019}. To do so, we restrict attention to the case with only one covariate, $Z$. After noting that $\Delta^2=r^2\var(Z)/\var(W)$ and rearranging, we can get
\begin{equation}
	\frac{\beta^2\var(Z)}{\var(W)} + \Delta\beta E_{\mathcal{W}_{\bm{\Delta}} }(\overline{%
 		\varepsilon }_{1}-\overline{\varepsilon }_{0}) > E_{\mathcal{W}_{\bm{\Delta}} }\left((\overline{\varepsilon }_{1}-\overline{\varepsilon }_{0})^2\right)\left(\frac{1}{ 1 - r^2}\right).
\end{equation}
The second term on the left-hand side disappears over random sampling. After replacing $\beta$ and $\varepsilon$ with structural parameters in a population model,\footnote{The model in \cite{mutz_perils_2019} can be written as $Y=\mu+\tau W + \gamma Z + \theta \xi$, where $\xi$ is an error term with unit variance.} it is possible to show that the inequality can be written as in \cite{mutz_perils_2019}:
\begin{equation}
	\frac{\gamma^2\var(Z)}{\theta^2} > \frac{1}{(n-1)(1-r^2)},
\end{equation}
where $\gamma$ is the effect of $Z$ on $Y$ and $\theta^2$ is the variance of the error term in the population model. \cite{mutz_perils_2019} use this equation to argue that, when a covariate is not very informative in explaining the outcome, one should be less inclined to control for $Z$ if $r^2$ is high. The reason if that an increase in $r^2$ increases the right-hand side of the equation. However, it is important to note that $\var(Z)$ is the sample variance and not the population variance. Over random sampling, an increase in $r^2$ may also increase $\var(Z)$ enough such that the inequality is \emph{more} likely to hold if $r^2$ increases. In general, it is not possible to say that we should be more or less inclined to control for covariates if we observe imbalances in covariates.

As noted by \cite{mutz_perils_2019}, conditional on $r^2$, the difference-in-means estimator is conditionally unbiased. This results hold exactly for any given sample. The reason is that the set containing all treatment assignments with a given R-squared, $\mathcal{W}_{r^2}$, must necessarily contain the mirrors of all assignments in the set. However, as shown in equation \eqref{eq:exp_dm}, this is not the case when conditioning on $\Delta$, the observed imbalance, which is what is typically shown in a table of balance tests. For instance, suppose one is interested in analyzing the effect of a vaccine in a randomized controlled trial, and there is a suspicion that the vaccine will be less effective among older individuals. If an imbalance is observed, such that the treatment group contains individuals that are on average one year older than the control group, it is not very helpful to note that the difference-in-means estimator is unbiased conditional on the treatment group containing individuals that are either one year older or one year younger than the control group. Instead, it makes sense to say that the difference-in-means estimator is biased conditional on the treatment group being one year older than the control group.







\section{Illustration \label{sec:illustration}}
To illustrate the results in the previous section, we perform a very simple simulation study where data is generated as $Y_i(0)= Z_i + u_i$ and $\tau=0$. To make it possible to go through all $n_A=\binom{n}{n_1}$ treatment assignments, we let $n=20$ and $n_1=10$. Both $Z$ and $u$ are drawn from a standard normal distribution. In total, we draw 1,000 samples and for each sample, we go through all $n_A=\binom{20}{10}= 184,756$ possible assignment vectors and calculate both $\hat\tau_{DM}$ and $\hat\tau_z$ for each of these vectors. In addition, we calculate the size of the statistical tests (conditional on $\bm\Delta$) as well as the conditional variance and MSE.

Figure \ref{fig:combined50} illustrates the results. Focusing on the point estimates, we see that the OLS estimator is conditionally unbiased (which means that over random sampling, $E_{\mathcal{W}_{\bm{\Delta}} }(\overline{\upvarepsilon}_1 - \overline{\upvarepsilon}_0)=0$), whereas the difference-in-means estimator is conditionally biased. Equation \eqref{eq:cond_e_reg_exact} implies that the bias should be $\bm\Delta'\bm\upbeta$ which in this case means that the bias should follow the 45 degree line. Indeed that is what is found. As we know should be the case, the unconditional expectations of the estimators are both exactly zero for each sample.

\begin{figure}
	\includegraphics[width=\linewidth]{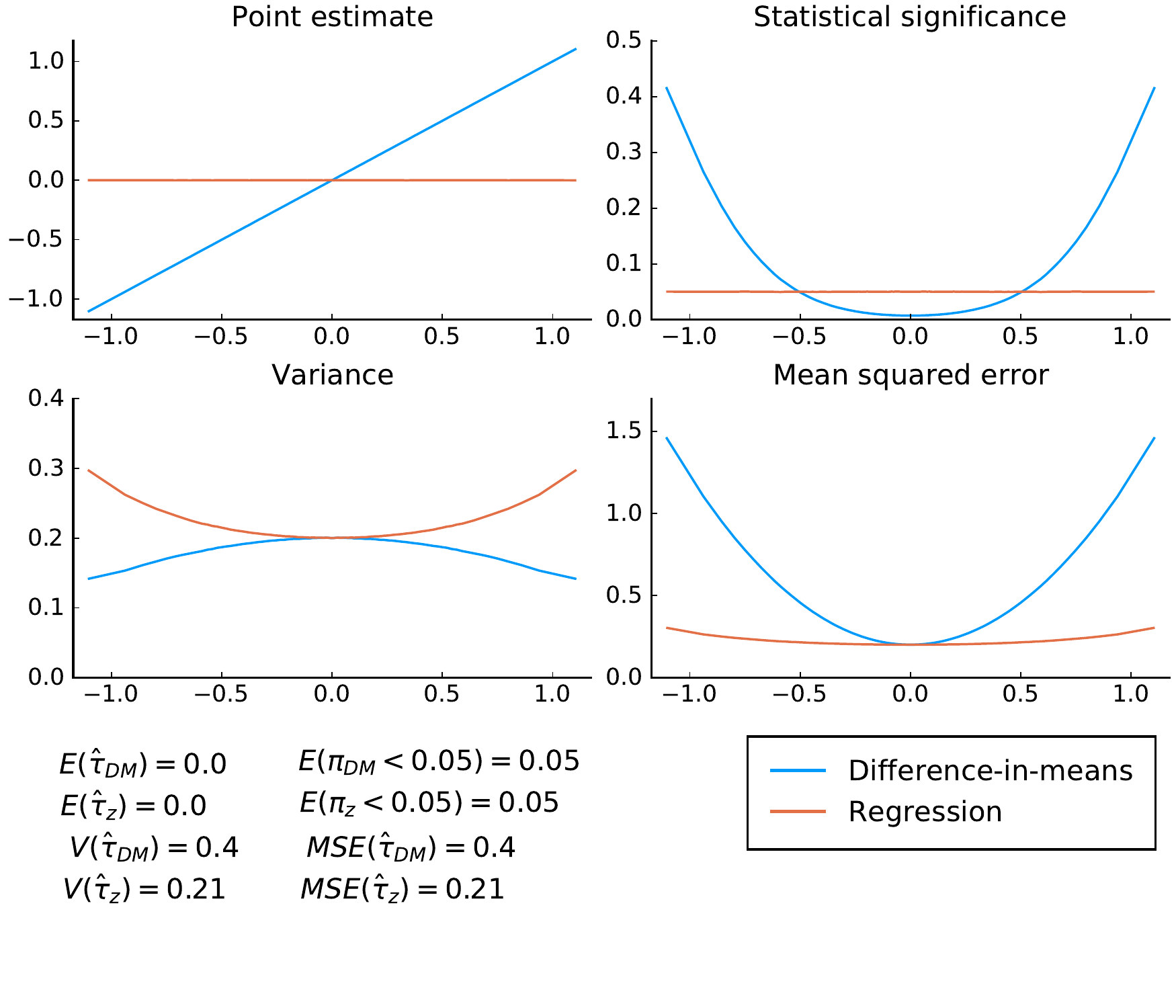} 
	\caption{Simulation results\label{fig:combined50}}
	\floatfoot{Note: The $x$-axis shows the average values of $\bm\Delta$ over the 1,000 replications for each percentile of $\bm\Delta$, whereas the $y$-axis indicate point estimate, statistical significance, variance and MSE for both the difference-in-means estimator and the OLS estimator. The unconditional values (i.e., not conditional on $\bm\Delta$) are shown in the bottom left corner of the figure. For the (two-sided) tests, the significance level is set at five percent. For the OLS estimators, the standard OLS covariance matrix is used. $\pi_{DM}$ and $\pi_{z}$ are the $p$-values from the respective tests.}
\end{figure}

Turning to the size of the tests, we first note that the unconditional size is correct for both estimators. The conditional test is correct for the OLS estimator, but wildly off for the difference-in-means estimator (a simple $t$-test). The more $\bm\Delta$ deviates from zero, the higher the rejection rate of the null hypothesis. Importantly, because the test has correct size on average, the size of the test conditional on $\bm\Delta$ being close to zero is smaller than 0.05, meaning the test in that range is conservative. It is also noteworthy that for no value of $\bm\Delta$ is the difference-in-means estimator conditionally unbiased with correct size of the hypothesis test.

The final two graphs show the conditional variance and MSE. Because the OLS estimator---but not the difference-in-means estimator---is conditionally unbiased, these are the same for the former but not the latter. The theoretical variances are given in equations \eqref{eq:var_cond_dm} and \eqref{eq:var_cond_reg}. The figure shows that $V_{\mathcal{W}_{\bm{\Delta}} }(\overline{\varepsilon }_{1}-\overline{\varepsilon }_{0})$ is decreasing as the magnitude of $\bm\Delta$ increases. The reason is that as $\bm\Delta$ increases, the assignment vectors become more similar to each other, and so $\overline\varepsilon_1-\overline\varepsilon_0$ become more similar. This result is in line with the theoretical result when the covariate is a dummy variable derived in the Appendix. For the OLS estimator, on the other hand, the term $\left( 1 - M_{\bm{\Delta}}/ (n-1) \right)^{-2}$ dominates such that the conditional variance is increasing in the magnitude of $\bm\Delta$. Consistent with the theoretical analysis, the conditional variance is identical between the two estimators when $\bm\Delta=\mathbf{0}$. The MSE is consistently greater for the difference-in-means estimator than the OLS estimator.

\section{Selection of covariates}
\subsection{Regression-based inference \label{sec:reg_based_inf}}
In situations when the number of observations are much larger than the number of relevant covariates ($n\gg K$), the preceding analysis suggests that it is always better to condition on the covariates than not condition on them as it will lead to a lower mean squared error and correct conditional inference. Even if a covariate is not relevant ($\beta=0$), little is lost with a large sample size. However, if $K$ is not order of magnitudes smaller than $n$, equation \eqref{eq:comp_dm_ols} implies that there is a tradeoff between adding more covariates as the bias term ($\bm{\Delta}^{\prime }\bm\upbeta$) decreases while the variance increases due to increase in the Mahalanobis distance, $M_{\bm\Delta}$. In the extreme case, with $K>n$, it is not even possible to condition on all covariates in a regression. So what should one do in such a case?

A common practice is to condition only on covariates which show large imbalances, but as \cite{mutz_perils_2019} show, such an approach will lead to incorrect inference. Another possibility would be to choose covariates based on perceived importance in explaining the outcome. However, unless such an approach is specified in a pre-analysis plan, it opens up the possibility for the researcher to select covariates in a large number of ways, potentially leading to issues such as data-mining and $p$-hacking. Even when the reseracher is completely honest, such an approach lack transparency, making it difficult for the research community at large to ascertain the credibility of the results.

It is therefore useful to have a rule-based system of covariate selection which limits the degrees of freedom of the researcher. We propose such a rule of covariate selection which builds on the idea of randomization inference. Randomization inference after covariate adjustments is conditional on a set of assignment vectors, $\mathcal{W}_{\bm\Delta}$ for which $\bm\Delta=\mathbf{c}$. If this set is too small, then randomization-based justification for inference collapses (\citealt{cox_randomization_2009}) and inference can only be justified under the assumption of random sampling from some population. The smallest $p$-value which can be attained from Fisher's exact test is $1/|\mathcal{W}_{\bm\Delta}|$, so. e.g., if it should be possible to achieve a $p$-value of 0.01 or smaller, it must be the case that there are at least 100 assignment vectors which has the same value of $\bm\Delta$.

If there are a few discrete covariates, then this would generally be true. However if the covariates are continuous, then it would typically be the case that $|\mathcal{W}_{\bm\Delta}|=1$ and, strictly speaking, inference based on the OLS estimator cannot be justified based on randomization.

Instead, we suggest basing inference on the set $\mathcal{W}_{\widetilde{\bm\Delta}}$ where all elements in the set yield a distance which is approximately equal to $\bm\Delta$. Note that, asymptotically, it is the case that $M_{\bm{\Delta}} \sim \chi^2(K)$. Let $\bm\Delta_j := \bm\Delta(\mathbf{W}_j) - \bm\Delta(\mathbf{W})$, where $\mathbf{W}$ is the assignment vector actually chosen and $\mathbf{W}_j\in\mathcal{W}$. It is the case that $M_{\bm\Delta_j}:=\frac{n_0n_1}{n}\bm\Delta_{j}'(\widetilde{\mathbf{Z}'}\mathbf{\widetilde{\mathbf{Z}}})^{-1}\bm\Delta_j$ follows a noncentral chi-square distribution with $K$ degrees of freedom and noncentrality parameter of $M_{\bm\Delta}$. We can now define the set $\mathcal{W}_{\widetilde{\bm\Delta}}$ as $\mathcal{W}_{\widetilde{\bm\Delta}} = \{ \mathbf{W}\in\mathcal{W} : M_{\bm\Delta_j} \leq \bar{\delta} \}$, where $\bar{\delta}$ is a small threshold value which should be set close to zero. For $\bar{\delta}=0$, it is the case that $\mathcal{W}_{\widetilde{\bm\Delta}}=\mathcal{W}_{\bm\Delta}$.

Let $H=|\mathcal{W}_{\widetilde{\bm\Delta}}|$ be the number of assignment vectors with small enough distance from the original treatment assignment to approximately justify randomization-based inference. In practice, for moderately sized $n$ it is not possible to go through all the $n_A=\binom{n}{n_1}$ assignment vectors to find $H$. However, by using the fact that $M_{\bm\Delta_j}$ follows a noncentral chi-square distribution, we can calculate the approximate size of the set as
\begin{equation}
	n_{\bar{\delta}} = F_{K,M_{\bm\Delta}}(\bar{\delta})  \cdot n_A,
\end{equation}
where $F_{K,M_{\bm\Delta}}(\cdot)$ is the cdf of the noncentral chi-square distribution with $K$ degrees of freedom and noncentrality parameter of $M_{\bm\Delta}$. If it is the case that $n_{\bar{\delta}} \geq H$, then the OLS estimator of the treatment effect, controlling for $\mathbf{Z}$, can be justified from a randomization inference perspective. In practice, if $\bar{\delta}$ is small and $K$ is reasonably large, it will be the case that $n_{\bar{\delta}} < H$. It is therefore necessary to somehow restrict the number of covariates that will be conditioned on.

We propose to condition on the principal components of $\mathbf{Z}$. There are two reasons for this proposal: First, if covariates are correlated, it is a natural way of reducing the dimensionality of the covariate space. Second, principal components are naturally ordered in descending variances. Let $\mathbf{Z}_p^{pc}=\left[\begin{array}{cccc} \mathbf{z}_1^{pc} & \mathbf{z}_2^{pc} & \hdots & \mathbf{z}_p^{pc}\end{array}\right]$ be a matrix of the first $p$ principal components of $\mathbf{Z}$; it is the case that $M^{pc}_{\bm{\Delta}}= \frac{n_0n_1}{n}\bm\Delta^{pc'}(\widetilde{\mathbf{Z}}_p^{pc'}\mathbf{\widetilde{\mathbf{Z}}}_p^{pc})^{-1}\bm\Delta^{pc}\sim \chi^2(p)$, resulting in $M^{pc}_{\bm{\Delta}_j}$ following a noncentral chi-square distribution with $p$ degrees of freedom and noncentrality parameter of $M^{pc}_{\bm{\Delta}}$.

With the natural ordering of the principal components, we suggest a simple algorithm (Algorithm \ref{alg:ComponentSelection}) which yield the number of principal componenst to condition on in a regression estimation of the treatment effect.

\begin{algorithm}
    \caption{Component selection \label{alg:ComponentSelection}}
    \begin{algorithmic}[1] 
    		\State Set $\bar{\delta}$ and $H$
            \State $p \gets 0$
            \State $n_{\bar{\delta}} \gets n_A$
            \While{$n_{\bar{\delta}} \geq H$}
            	\State Select first $p+1$ principal components and calculate the Mahalanobis distance
            	\State Calculate $n_{\bar{\delta}}$
                \If{$n_{\bar{\delta}} \geq H$}
                $p \gets p+1$
                \EndIf
            \EndWhile
            \State \textbf{return} $p$
    \end{algorithmic}
\end{algorithm}

After the components have been selected, we get the treatment effect estimator from a regression of $Y$ on the treatment indicator, controlling for the $p$ principal components.

\subsection{Simulation results \label{sec:sim_res}}
To study how our algorithm compares to other estimators, we perform a simple simulation study. Specifically, we generate data as
\begin{equation}
	Y_i(0) = Y_i(1) = \mathbf{z}_i\mathbf{b} + u_i,
\end{equation}
where $\mathbf{Z}\sim N(\mathbf{0}, \mathbf{I})$, $\mathbf{b}=\left[\begin{array}{cccc}
	\frac{1}{\sqrt{K}} & \frac{1}{\sqrt{K}} & \ldots & \frac{1}{\sqrt{K}} \\
\end{array}\right]'$ and $u\sim N(0,1)$. With this setup, we have $\var(\mathbf{z}_i\mathbf{b})=\var(u_i)$, which means the $R^2$ from a regression of $\mathbf{Y(0)}$ on $\mathbf{Z}$ should be around 0.5. For a randomly selected sample, we draw 10,000 random treatment assignment vectors and estimate the treatment effect. We then repeat this process for 1,000 different samples and calculate the average MSE. For our algorithm, we let $\bar{\delta}=0.01$ and the sample size is set to $n=50$ with $n_0=n_1=25$. We vary $K$ (the number of covariates) from 2 to 40 in steps of 2.

With this setup, the covariates are orthogonal to each other in the population, and so we should not expect the PCA to effectively reduce the dimensionality of the data. Hence, this setup can be considered a ``worst case'' for our method. To study what happens when covariates are correlated, we use the method suggested by \cite{lewandowski_generating_2009} to generate correlated covariates with the parameter $\eta$ being set to one.

We contrast our estimator with three other estimators: i) the difference-in-means estimator, ii) the OLS estimator when all covariates are used as controls and iii) the cross-estimation estimator suggested by \cite{wager_high-dimensional_2016}. The latter estimator uses high-dimensional regression adjustments with an elastic net to select important covariates when there are many covariates relative to the number of observations.

Figure \ref{fig:sim_results_2_40} shows the result from the simulations. Beginning with the left graph---which shows the results from orthogonal covariates---we see that with few covariates, the MSE of the difference-in-means estimator is around double that of the OLS estimator, which is what we should expect for $n\gg K$ as the covariates account for fifty percent of the variation in $Y(0)$ (see, for instance, \citealt{morgan_rerandomization_2012}). Notably, the OLS estimator and our PCA-based estimator is identical in that case. The reason is simply that with so few covariates, all principal components are selected, and conditioning on all principal components is equivalent to conditioning on all covariates. The cross-estimation estimator lies somewhere between the difference-in-means estimator and the other estimators.

\begin{figure}
	\captionsetup{position=top}
	\caption{MSE, homogeneous treatment effect \label{fig:sim_results_2_40}}
	\includegraphics[width=\linewidth]{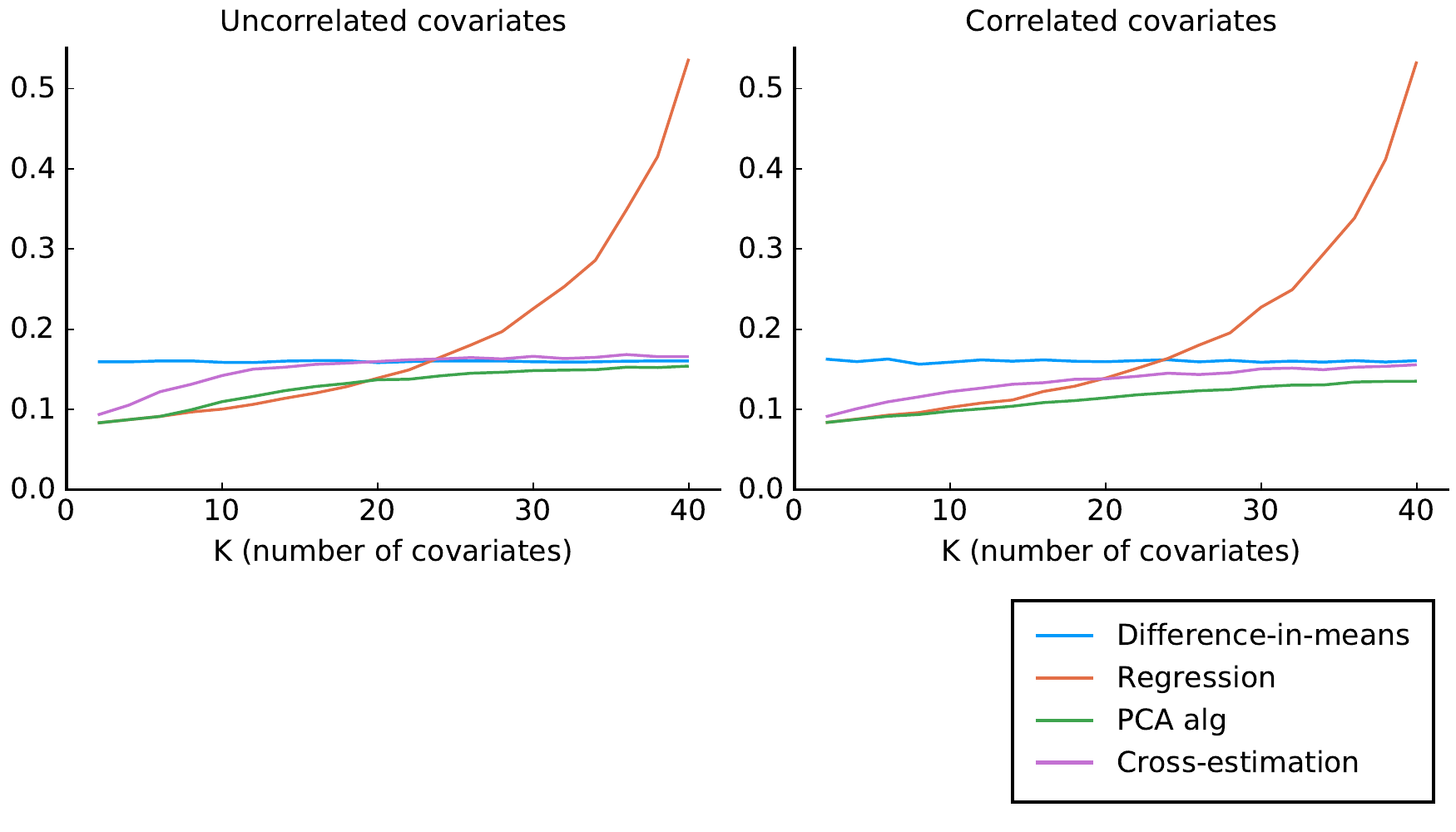}
	\floatfoot{\footnotesize Note: For each value of $K$, 1,000 samples are drawn with 10,000 assignment vectors selected for each sample. For the cross-estimation estimator, for computational time purposes, only 100 assignment vectors are selected for each sample. The average MSE is shown for each $K$. The sample size is set to 50 and $\tau=0$.}
\end{figure}

As $K$ increses, the MSE of the difference-in-means estimator is naturally unchanged, while the MSE of the three other estimators increases. For an interval with $K$ between 10 and 20, the OLS estimator marginally outperforms our estimator, but once the number of covariates increases further, the MSE of the OLS estimator skyrockets. For our estimator, the MSE increases slowly and stays consistently lower than that of the difference-in-means estimator. The cross-estimation estimator is clearly better than the OLS estimator for large $K$, but performs worse than our estimator.

The left graph shows the results from the worst case for our estimator. In the right graph, we show results when covariates are correlated. The difference-in-means and OLS estimators are very similar to the previous case, but now our estimator outperforms both of them for all values of $K$ (except for small $K$ when the OLS estimator and our estimator are equivalent). The cross-estimation estimator also outperforms the other estimators for large $K$, but still performs worse than our estimator.

The results in Figure \ref{fig:sim_results_2_40} shows the average of the MSE for each value of $K$. However, as we discuss previously, the MSE will depend on $\bm\Delta$. Because $\bm\Delta$ is $K$-dimensional, it is not possible to illustrate the results as we did in Figure \ref{fig:combined50}. Instead, for each sample, we take the average MSE of each percentile of the Mahalanobis distance, $M_{\bm{\Delta}}$, and then take the average for each percentile over all 1,000 samples. We show the results for $K=10,20,30$.

Results are shown in Figure \ref{fig:mse_cond_M} with the uncorrelated covariates in the top panel and the correlated covariates in the bottom panel. In Figure \ref{fig:combined50}, the MSE displayed a U-shaped pattern with minimum when $\bm\Delta=0$. Because the Mahalanobis distance is a (weighted) square of $\bm\Delta$, the MSE is now increasing in the Mahalanobis distance for all four estimators. We see that the MSE of the difference-in-means estimator, our PCA estimator and the cross-estimation estimator all increase at roughly the same pace, while the OLS estimator has an MSE that increases sharply for large distances once $K$ is large. Note that the cross-estimation estimator is more variable because, for computational time purposes, we only selected 100 instead of 10,000 random assignment vectors per random sample.

\begin{figure}
	\captionsetup{position=top}
	\caption{MSE by percentile of the original Mahalanobis distance \label{fig:mse_cond_M}}
	\subfloat[Uncorrelated covariates]{
		\includegraphics[width=\linewidth]{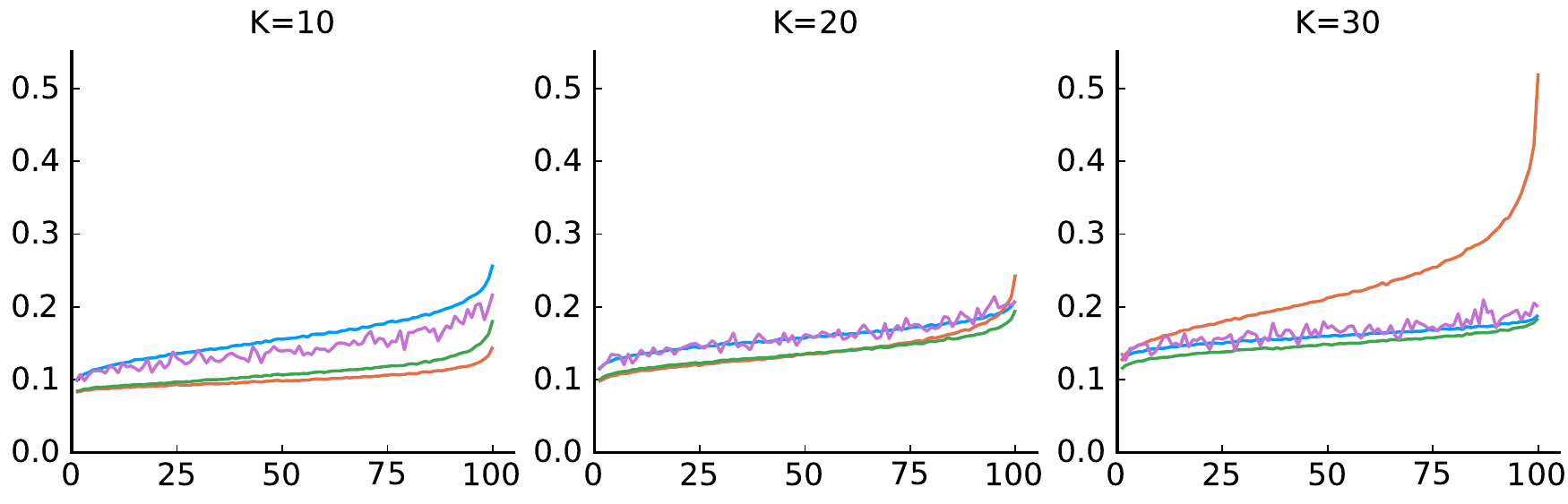}}\\
	\subfloat[Correlated covariates]{
		\includegraphics[width=\linewidth]{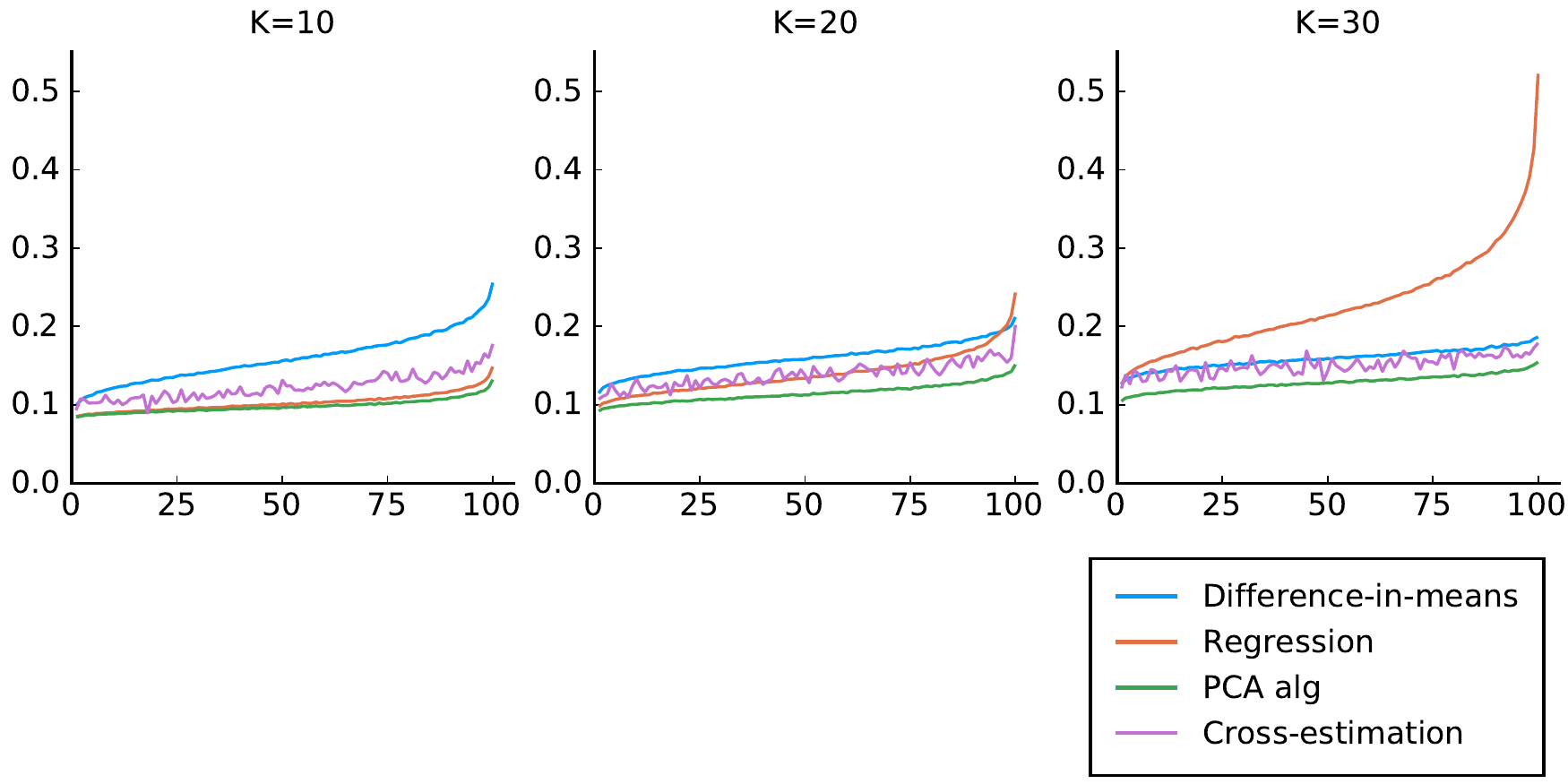}}
	\floatfoot{\footnotesize Note: For each value of $K$, 1,000 samples are drawn with 10,000 assignment vectors selected for each sample. For the cross-estimation estimator, for computational time purposes, only 100 assignment vectors are selected for each sample. The sample size is set to 50 and $\tau=0$.}
\end{figure}

To help with the interpretation of the results, Figure \ref{fig:fig_nr_p} shows the number of principal components selected for our PCA estimator by Algorithm \ref{alg:ComponentSelection} for different values of $K$. The figure shows results only for the uncorrelated covariates, but the results are virtually identical with correlated covariates. Overall, the number of components decreases as the original Mahalanobis distance increases, as it is, on average, more difficult to find similar components in that case.

\begin{figure}
	\captionsetup{position=top}
	\caption{Number of components by percentile of the original Mahalanobis distance \label{fig:fig_nr_p}}
	\includegraphics[width=.7\linewidth]{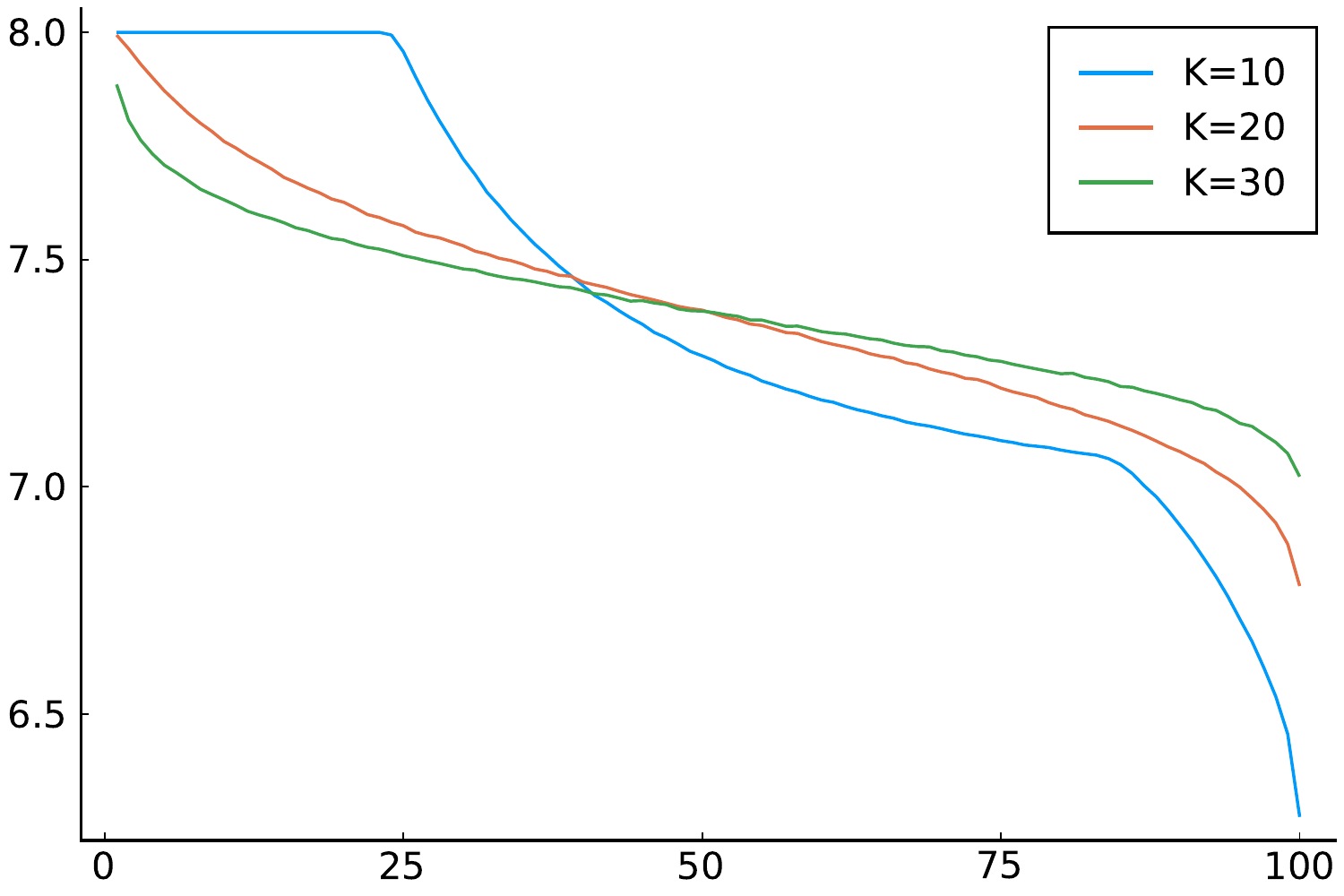}
	\floatfoot{\footnotesize Note: The figure shows the average number of components selected by Algorithm \ref{alg:ComponentSelection} for the uncorrelated covariates. With correlated covariates, the figure would be virtually identical.}
\end{figure}

In Table \ref{tab:size_homo} we show the size of a two-sided test of $\tau=0$ for each of the four estimators (five percent significance level) for $K=10,20,30$. The first column shows the average size (independent of the Mahalanobis distance). As can be seen, the difference-in-means estimator, the OLS estimator and our PCA estimator all have approximately correct size, whereas the cross-estimation estimator overrejects the null, with a rejection rate of around eight percent instead of five.

\begin{table}[p!]
	\begin{threeparttable}
	    \caption{Size, homogeneous effects \label{tab:size_homo}}
	    \begin{tabular*}{\textwidth}{l @{\extracolsep{\fill}} cccccc}
\toprule
Quintiles & All & 1st & 2nd & 3rd & 4th & 5th \\
\midrule & \multicolumn{6}{c}{Uncorrelated covariates} \\
\cmidrule(lr){2-7}
\multicolumn{7}{l}{\(K=10\)}\\
Difference-in-means & 0.05 & 0.023 & 0.036 & 0.047 & 0.06 & 0.085 \\
Cross-estimation & 0.08 & 0.053 & 0.067 & 0.073 & 0.09 & 0.116 \\
Regression & 0.05 & 0.05 & 0.05 & 0.05 & 0.05 & 0.05 \\
PCA alg & 0.051 & 0.044 & 0.047 & 0.05 & 0.053 & 0.064 \\
\addlinespace
\addlinespace\multicolumn{7}{l}{\(K=20\)}\\
Difference-in-means & 0.05 & 0.032 & 0.042 & 0.049 & 0.057 & 0.07 \\
Cross-estimation & 0.082 & 0.058 & 0.072 & 0.079 & 0.089 & 0.111 \\
Regression & 0.05 & 0.05 & 0.05 & 0.05 & 0.05 & 0.05 \\
PCA alg & 0.051 & 0.038 & 0.045 & 0.05 & 0.056 & 0.066 \\
\addlinespace
\addlinespace\multicolumn{7}{l}{\(K=30\)}\\
Difference-in-means & 0.05 & 0.038 & 0.045 & 0.05 & 0.055 & 0.062 \\
Cross-estimation & 0.08 & 0.064 & 0.072 & 0.08 & 0.085 & 0.1 \\
Regression & 0.05 & 0.05 & 0.05 & 0.05 & 0.05 & 0.05 \\
PCA alg & 0.051 & 0.041 & 0.047 & 0.051 & 0.055 & 0.061 \\
\addlinespace
\addlinespace & \multicolumn{6}{c}{Correlated covariates} \\
\cmidrule(lr){2-7}
\multicolumn{7}{l}{\(K=10\)}\\
Difference-in-means & 0.05 & 0.025 & 0.036 & 0.047 & 0.059 & 0.083 \\
Cross-estimation & 0.076 & 0.056 & 0.064 & 0.073 & 0.084 & 0.101 \\
Regression & 0.05 & 0.05 & 0.05 & 0.05 & 0.05 & 0.05 \\
PCA alg & 0.051 & 0.049 & 0.05 & 0.05 & 0.051 & 0.053 \\
\addlinespace
\addlinespace\multicolumn{7}{l}{\(K=20\)}\\
Difference-in-means & 0.05 & 0.032 & 0.042 & 0.049 & 0.057 & 0.07 \\
Cross-estimation & 0.079 & 0.061 & 0.071 & 0.079 & 0.086 & 0.101 \\
Regression & 0.05 & 0.05 & 0.05 & 0.05 & 0.05 & 0.05 \\
PCA alg & 0.051 & 0.044 & 0.048 & 0.051 & 0.054 & 0.06 \\
\addlinespace
\addlinespace\multicolumn{7}{l}{\(K=30\)}\\
Difference-in-means & 0.05 & 0.038 & 0.045 & 0.05 & 0.055 & 0.062 \\
Cross-estimation & 0.082 & 0.068 & 0.079 & 0.081 & 0.084 & 0.097 \\
Regression & 0.05 & 0.05 & 0.05 & 0.05 & 0.05 & 0.05 \\
PCA alg & 0.051 & 0.044 & 0.048 & 0.051 & 0.054 & 0.058 \\
\addlinespace
\addlinespace\bottomrule
\end{tabular*}
	    \begin{tablenotes}
	    	\item[] \footnotesize Note: The table shows the size of a two-sided test of $\tau=0$ at five percent significance level. The first column shows the unconditional size, whereas the next five shows the size for each quintile of the Mahalanobis distance, $M_{\bm\Delta}$. For each value of $K$, 1,000 samples are drawn with 10,000 assignment vectors selected for each sample. For the cross-estimation estimator, for computational time purposes, only 100 assignment vectors are selected for each sample. The sample size is set to 50. For the regression-based estimators, the standard OLS covariance matrix is used.
	    \end{tablenotes}
	\end{threeparttable}
\end{table}

The following columns show the results separately for each quintile of the Mahalanobis distance, $M_{\bm\Delta}$. We now see that only the OLS estimator maintains correct size regardless of the value of the Mahalanobis distance, whereas the difference-in-means estimator clearly underrejects for small values of the Mahalanobis distance and overrejects for large values. This pattern is expected, as the difference-in-means estimator does not take the covariate imbalance into account. A similar pattern is found for the cross-estimation estimator, but with a higher rejection rate. Finally, for our PCA estimator, the rejection rate is also increasing with the Mahalanobis distance, but at a slower pace, as the covariate imbalance is partially taken into account by the selected principal components.

Finally, Table \ref{tab:power_homo} shows the power of the different estimators with $\tau$ set to one (from a two-sided test of $\tau=0$). The results are very similar to the result for the MSE: with $K=10$, the OLS estimator is the most powerful estimator, closely followed by our PCA estimator. For larger $K$, the OLS estimator becomes much worse, while the PCA estimator continuous to perform well. The cross-estimation estimator also performs comparatively well for $K=30$, but it should be noted that the power is not size-adjusted.

\begin{table}[p!]
	\begin{threeparttable}
	    \caption{Power, homogeneous effects \label{tab:power_homo}}
	    \begin{tabular*}{\textwidth}{l @{\extracolsep{\fill}} cccccc}
\toprule
Quintiles & All & 1st & 2nd & 3rd & 4th & 5th \\
\midrule & \multicolumn{6}{c}{Uncorrelated covariates} \\
\cmidrule(lr){2-7}
\multicolumn{7}{l}{\(K=10\)}\\
Difference-in-means & 0.686 & 0.709 & 0.695 & 0.686 & 0.677 & 0.664 \\
Cross-estimation & 0.792 & 0.815 & 0.802 & 0.797 & 0.781 & 0.768 \\
Regression & 0.858 & 0.899 & 0.881 & 0.865 & 0.845 & 0.801 \\
PCA alg & 0.839 & 0.882 & 0.862 & 0.843 & 0.823 & 0.785 \\
\addlinespace
\addlinespace\multicolumn{7}{l}{\(K=20\)}\\
Difference-in-means & 0.686 & 0.701 & 0.691 & 0.686 & 0.68 & 0.673 \\
Cross-estimation & 0.74 & 0.759 & 0.744 & 0.741 & 0.737 & 0.72 \\
Regression & 0.744 & 0.831 & 0.788 & 0.753 & 0.713 & 0.635 \\
PCA alg & 0.759 & 0.797 & 0.773 & 0.758 & 0.743 & 0.72 \\
\addlinespace
\addlinespace\multicolumn{7}{l}{\(K=30\)}\\
Difference-in-means & 0.689 & 0.699 & 0.692 & 0.688 & 0.685 & 0.681 \\
Cross-estimation & 0.734 & 0.745 & 0.743 & 0.732 & 0.728 & 0.72 \\
Regression & 0.545 & 0.682 & 0.604 & 0.55 & 0.493 & 0.398 \\
PCA alg & 0.724 & 0.75 & 0.733 & 0.723 & 0.714 & 0.7 \\
\addlinespace
\addlinespace & \multicolumn{6}{c}{Correlated covariates} \\
\cmidrule(lr){2-7}
\multicolumn{7}{l}{\(K=10\)}\\
Difference-in-means & 0.698 & 0.717 & 0.705 & 0.698 & 0.69 & 0.678 \\
Cross-estimation & 0.838 & 0.854 & 0.84 & 0.839 & 0.837 & 0.818 \\
Regression & 0.865 & 0.905 & 0.888 & 0.872 & 0.852 & 0.808 \\
PCA alg & 0.882 & 0.91 & 0.896 & 0.885 & 0.872 & 0.846 \\
\addlinespace
\addlinespace\multicolumn{7}{l}{\(K=20\)}\\
Difference-in-means & 0.691 & 0.703 & 0.695 & 0.691 & 0.686 & 0.679 \\
Cross-estimation & 0.801 & 0.816 & 0.807 & 0.797 & 0.797 & 0.788 \\
Regression & 0.732 & 0.82 & 0.776 & 0.742 & 0.701 & 0.622 \\
PCA alg & 0.82 & 0.85 & 0.832 & 0.821 & 0.808 & 0.787 \\
\addlinespace
\addlinespace\multicolumn{7}{l}{\(K=30\)}\\
Difference-in-means & 0.684 & 0.692 & 0.687 & 0.684 & 0.68 & 0.676 \\
Cross-estimation & 0.769 & 0.781 & 0.774 & 0.766 & 0.766 & 0.757 \\
Regression & 0.526 & 0.661 & 0.584 & 0.53 & 0.473 & 0.381 \\
PCA alg & 0.779 & 0.802 & 0.787 & 0.778 & 0.769 & 0.757 \\
\addlinespace
\addlinespace\bottomrule
\end{tabular*}
	    \begin{tablenotes}
	    	\item[] \footnotesize Note: The table shows the power from of a two-sided test of $\tau=0$ at five percent significance level, with $\tau=1$. The first column shows the unconditional power, whereas the next five shows the power for each quintile of the Mahalanobis distance, $M_{\bm\Delta}$. For each value of $K$, 1,000 samples are drawn with 10,000 assignment vectors selected for each sample. For the cross-estimation estimator, for computational time purposes, only 100 assignment vectors are selected for each sample. The sample size is set to 50. For the regression-based estimators, the standard OLS covariance matrix is used.
	    \end{tablenotes}
	\end{threeparttable}
\end{table}

Overall, we conclude that our PCA-based estimator generally outperforms the other three estimator in terms of MSE. While the size is not always correct conditional on observed differences in covariates, the issue is smaller than for the difference-in-means estimator or cross-estimation estimator. It is also important to note that when $n\gg K$, our estimator essentially collapses to the OLS estimator.

\subsection{Randomization inference}
The OLS estimator above builds on the idea that---in principle---there are a number of different treatment assignments which give approximately the same Mahalanobis distance between treatment and control in the selected principal components. However, we do not actually find these assignments, and the estimator is a standard OLS estimator.

In this section, we instead suggest that there is a way to perform randomization inference in the spirit of Fisher's exact test, but conditional on the observed distance, $\bm\Delta$. By the same argument as for regression-based inference, we focus on the set $\mathcal{W}_{\widetilde{\bm\Delta}}$. Under the sharp null, all data are observed and we can calculate the treatment effect estimates for all assignment vectors in this set. The $p$-value is retrieved in the same way as for Fisher's exact test, with the difference that we consider the set $\mathcal{W}_{\widetilde{\bm\Delta}}$ instead of $\mathcal{W}$: sort the absolute values of the difference-in-means estimator from the assignment vectors $\mathcal{W}_{\widetilde{\bm\Delta}}$ in descending order. Let $r$ be the rank of the estimator for assignment vector $\mathbf{W}$, the $p$-value is given by $r/H$.

If $\bar{\delta}=0$, this test is exact. However, when $\bar{\delta}$ is small, the test will only have approximately correct size. For small sample sizes, it is possible to go through all assignment vectors to identify the set $\mathcal{W}_{\widetilde{\bm\Delta}}$. However, for larger sample sizes, this is no longer possible, and the $p$-value will have to be Monte Carlo approximated by sampling a subset of $\mathcal{W}_{\widetilde{\bm\Delta}}$.

To illustrate the properties of this test in the former case (i.e., when it is possible to go through all the assignment vectors)---and to contrast it with other versions of Fisher's exact test---we perform the same exercise as the one in Section \ref{sec:illustration} with $Y_i(0)= Z_i + u_i$, $\tau=0$ and $n_0=n_1=10$, where both $Z$ and $u$ follow a standard normal distribution. Once again, because the sample size is small, we can go through all $n_A=184,756$ treatment assignment vectors.

To perform the test, we take the given treatment assignment, form the set $\mathcal{W}_{\widetilde{\bm\Delta}}$ with $\bar\delta=0.01$, and compare the rank of the treatment effect estimate compared to all other estimates formed by the assignment vectors in the set. Note that the cardinality of $\mathcal{W}_{\widetilde{\bm\Delta}}$ will vary, and therefore also the resolution of the $p$-values. For small magnitudes of $\bm\Delta$ the resolution will generally be high, whereas for the most extreme treatment assignments (i.e., when the magnitude of $\bm\Delta$ is the largest), it is even possible that the set only contain the original treatment assignment itself.

We contrast this test with i) Fisher's exact test and ii) the regression-based Fisher-test. In the latter, we perform $n_A$ different regressions and compare the estimate of the actual treatment assignment with all other estimates.  Note that both Fisher's exact test and the regression-based test are exact. That is, the null is rejected in exactly 9,237 out of 184,756 cases in every sample.

Results are shown in Figure \ref{fig:size_cond_delta} where we have performed the analysis for 1,000 different samples, and just as in Figure \ref{fig:combined50}, we illustrate the size conditional on $\bm\Delta$. For Fisher's exact test, as expected, we see virtually the same pattern as for the $t$-test. That is, when $\bm\Delta$ is at the extreme ends, the null is severely overrejected, whereas when $\bm\Delta$ is close to zero, the reverse happens.

Perhaps more surprisingly is that the size is incorrect for the regression-based Fisher-test conditional on $\bm\Delta$. This is different from the corresponding regression-based test in the Neyman-Pearson framework. While the pattern is less stark than for Fisher's exact test, once again we see overrejection when the magnitude of $\bm\Delta$ is large and underrejection when it is small.

For our approximate Fisher-test on the other hand, the size is approximately correct even conditional on $\bm\Delta$. While we cannot be guaranteed exact correct size over all possible assignment vectors, we can see that it holds approximately.

\begin{figure}
	\captionsetup{position=top}
	\caption{Size conditional on $\bm\Delta$ \label{fig:size_cond_delta}}
	\includegraphics[width=.75\linewidth]{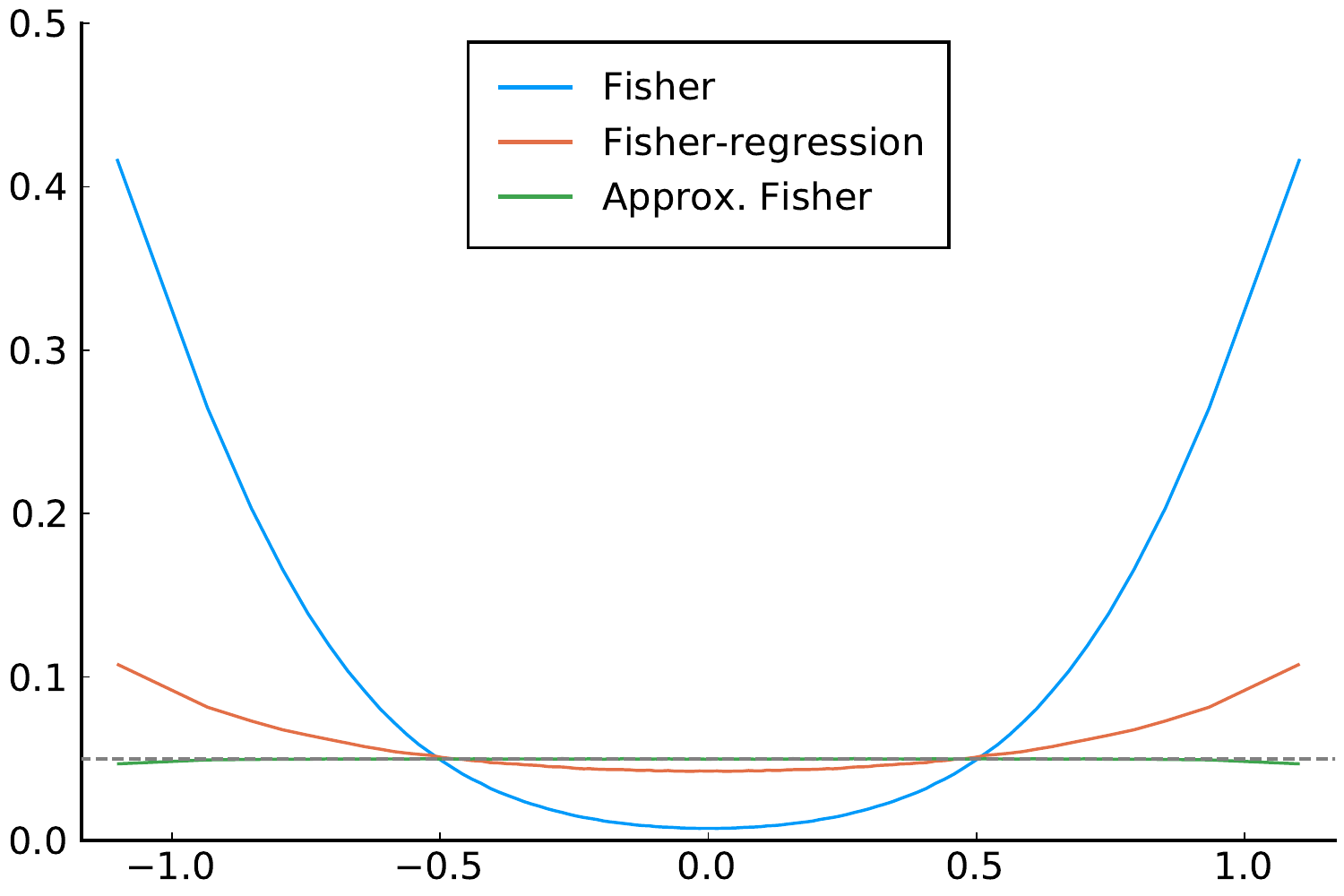}
	\floatfoot{\footnotesize Note: The $x$-axis shows the average values of $\bm\Delta$ over the 1,000 replications for each percentile of $\bm\Delta$, whereas the $y$-axis indicate statistical significance for Fisher's exact test, the regression-based Fisher-test and our approximate Fisher-test at five percent significance level. The dashed line indicates the correct size at .05.}
\end{figure}

In this simple illustration of the test, we only have one covariate. In the general case with $K$ covariates, for the same reason as outlined in Section \ref{sec:reg_based_inf}, we focus on the principal components. However, we cannot use Algorithm \ref{alg:ComponentSelection} to find the number of components for performing the test, because we now have to actually find all the assignment vectors. Let $p_F$ be the number of components selected for the randomization test. It is the case that $p_F\leq p$.

Depending on how one searches for $\mathcal{W}_{\widetilde{\bm\Delta}}$, the number of components will likely vary. If $\mathcal{W}_{\widetilde{\bm\Delta}}$ is found by rerandomization, i.e., by randomly sampling assignment vectors from $\mathcal{W}$ until a sufficient number of assignment vectors (such as 100) have been found, then the number of components will likely be very few, as rerandomization is a fairly inefficient way of finding similar assignment vectors, especially for relatively large $p$ when $|\mathcal{W}_{\widetilde{\bm\Delta}}| \ll |\mathcal{W}|$.

A much more efficient algorithm is the greedy pair-switching algorithm of \cite{krieger_nearly_2019}. In that algorithm, a random initial assignment vector is selected. One then searches through all $n_0n_1$ possible pair-switches of treatment and control and switches treatment status for the pair-switch which decreases the Mahalanobis distance, $M_{\bm\Delta_j}$, the most. The resulting assignment vector becomes the new initial vector from which it is possible to search $n_0n_1$ new pair-switches. This process continues until there is no pair-switch that can be made to further decrease the Mahalanobis distance, in which case the algorithm is finished. The algorithm can then be repeated for a new random initial vector.

With the greedy pair-switching algorithm, the issue remains on how many components should be selected. We propose a very simple rule which begins with first getting the number of components using Algorithm \ref{alg:ComponentSelection}. We then run the pair-switching algorithm a decently large number of times, say $n_s=1,000$. If there are at least a small number of assignment vectors, $n_f$, out of the $n_s$ resulting assignment vectors which satisfy $M_{\bm\Delta_j} \leq \bar{\delta}$, we let the algorithm run until $H$ assignment vectors are found for which $M_{\bm\Delta_j} \leq \bar{\delta}$. If not, we repeat the search but with the first $p-1$ components instead. What values of $n_f$ and $n_s$ to choose depend on the computational time available.

We perform a Monte Carlo analysis with $n_1=n_0=25$ and $K=10,20,30$, with the results being shown in Table \ref{tab:size_homo_fisher} for four different Fisher tests, where we include two different regression based-based tests: one that is performed on the same components as our approximate Fisher-test and one which instead condition on all covariates, $\mathbf{Z}$. For our approximate Fisher-test, we set $n_s=1,000$ and $n_f=20$.

\begin{table}[p!]
\begin{threeparttable}
    \caption{Size, homogeneous effects, Fisher tests \label{tab:size_homo_fisher}}
    \begin{tabular*}{\textwidth}{l @{\extracolsep{\fill}} cccccc}
\toprule
Quintiles & All & 1st & 2nd & 3rd & 4th & 5th \\
\midrule & \multicolumn{6}{c}{Uncorrelated covariates} \\
\cmidrule(lr){2-7}
\multicolumn{7}{l}{\(K=10\)}\\
Approx. Fisher (PCA) & 0.049 & 0.037 & 0.044 & 0.047 & 0.053 & 0.065 \\
Fisher-regression (All) & 0.049 & 0.036 & 0.042 & 0.047 & 0.054 & 0.068 \\
Fisher-regression (PCA) & 0.049 & 0.03 & 0.041 & 0.047 & 0.056 & 0.073 \\
Fisher & 0.049 & 0.024 & 0.035 & 0.047 & 0.061 & 0.08 \\
\addlinespace
\addlinespace\multicolumn{7}{l}{\(K=20\)}\\
Approx. Fisher (PCA) & 0.05 & 0.036 & 0.046 & 0.052 & 0.052 & 0.065 \\
Fisher-regression (All) & 0.049 & 0.026 & 0.034 & 0.046 & 0.057 & 0.084 \\
Fisher-regression (PCA) & 0.05 & 0.03 & 0.044 & 0.049 & 0.057 & 0.068 \\
Fisher & 0.05 & 0.031 & 0.044 & 0.049 & 0.056 & 0.069 \\
\addlinespace
\addlinespace\multicolumn{7}{l}{\(K=30\)}\\
Approx. Fisher (PCA) & 0.05 & 0.041 & 0.046 & 0.05 & 0.054 & 0.061 \\
Fisher-regression (All) & 0.05 & 0.016 & 0.029 & 0.042 & 0.058 & 0.105 \\
Fisher-regression (PCA) & 0.05 & 0.038 & 0.042 & 0.048 & 0.055 & 0.066 \\
Fisher & 0.049 & 0.039 & 0.045 & 0.048 & 0.053 & 0.06 \\
\addlinespace
\addlinespace & \multicolumn{6}{c}{Correlated covariates} \\
\cmidrule(lr){2-7}
\multicolumn{7}{l}{\(K=10\)}\\
Approx. Fisher (PCA) & 0.05 & 0.045 & 0.049 & 0.05 & 0.052 & 0.056 \\
Fisher-regression (All) & 0.05 & 0.037 & 0.042 & 0.047 & 0.052 & 0.074 \\
Fisher-regression (PCA) & 0.05 & 0.038 & 0.044 & 0.05 & 0.054 & 0.066 \\
Fisher & 0.05 & 0.025 & 0.037 & 0.048 & 0.059 & 0.082 \\
\addlinespace
\addlinespace\multicolumn{7}{l}{\(K=20\)}\\
Approx. Fisher (PCA) & 0.05 & 0.041 & 0.046 & 0.049 & 0.053 & 0.061 \\
Fisher-regression (All) & 0.05 & 0.025 & 0.036 & 0.046 & 0.058 & 0.084 \\
Fisher-regression (PCA) & 0.049 & 0.034 & 0.042 & 0.048 & 0.055 & 0.065 \\
Fisher & 0.05 & 0.032 & 0.042 & 0.049 & 0.055 & 0.068 \\
\addlinespace
\addlinespace\multicolumn{7}{l}{\(K=30\)}\\
Approx. Fisher (PCA) & 0.051 & 0.041 & 0.05 & 0.052 & 0.052 & 0.062 \\
Fisher-regression (All) & 0.049 & 0.016 & 0.029 & 0.042 & 0.059 & 0.102 \\
Fisher-regression (PCA) & 0.05 & 0.036 & 0.046 & 0.05 & 0.052 & 0.064 \\
Fisher & 0.05 & 0.036 & 0.048 & 0.051 & 0.051 & 0.065 \\
\addlinespace
\addlinespace\bottomrule
\end{tabular*}
    \begin{tablenotes}
	    \item[] \footnotesize Note: The table shows the size from of a test of the sharp null at five percent significance level. The first column shows the unconditional size, whereas the next five shows the size for each quintile of the Mahalanobis distance, $M_{\bm\Delta}$. For each value of $K$, 1,000 samples are drawn with 100 assignment vectors selected for each sample. The sample size is set to 50 and the following parameters are used: $n_s=1000$, $n_f=20$ and $\bar{\delta}=0.01$.
	\end{tablenotes}
\end{threeparttable}
\end{table}

The first thing to note is that all four tests have the correct size unconditionally. Conditional on the Mahalanobis distance, the picture is different. For small Mahalanobis distances, all tests underrejects the null, whereas for large Mahalanobis distances, they all overreject the null. In all cases, however, the approximate Fisher is closest to having correct conditional size. 

Why does our approximate Fisher-test not have correct conditional size in line with the results in Figure \ref{fig:size_cond_delta}? The reason is that we do not condition on all principal components. If we were to do so (which can happen with few covariates), the test would maintain approximately correct size also conditional on the Mahalanobis distance.


Table \ref{tab:power_homo_fisher} shows the power of the four tests when $\tau=1$. On average, we can see that the regression-based Fisher-test using the using principal components is the overall most powerful test (the exception of $K=10$ and uncorrelated covariates when the regression-based Fisher-test is the most powerful). The power of the approximate Fisher-test is close to that of the principal component regression-based Fisher-test for large $K$. As expected, the power of the standard Fisher-test is, basically, constant across $K$. It is also the overall least powerful test for all $K$ less than 30. The power of the regression-based Fisher-test decreases sharply with $K$ and when $K=30$, the power is substantially lower than for all the other tests. This result clearly shows the importance of not naively controlling for covariates using regression when conducting randomized inference. 

\begin{table}[p!]
\begin{threeparttable}
    \caption{Power, homogeneous effects, Fisher tests \label{tab:power_homo_fisher}}
    \begin{tabular*}{\textwidth}{l @{\extracolsep{\fill}} cccccc}
\toprule
Quintiles & All & 1st & 2nd & 3rd & 4th & 5th \\
\midrule & \multicolumn{6}{c}{Uncorrelated covariates} \\
\cmidrule(lr){2-7}
\multicolumn{7}{l}{\(K=10\)}\\
Approx. Fisher (PCA) & 0.738 & 0.788 & 0.751 & 0.734 & 0.722 & 0.696 \\
Fisher-regression (All) & 0.847 & 0.857 & 0.85 & 0.848 & 0.842 & 0.835 \\
Fisher-regression (PCA) & 0.79 & 0.811 & 0.796 & 0.792 & 0.784 & 0.768 \\
Fisher & 0.679 & 0.705 & 0.682 & 0.678 & 0.672 & 0.658 \\
\addlinespace
\addlinespace\multicolumn{7}{l}{\(K=20\)}\\
Approx. Fisher (PCA) & 0.705 & 0.731 & 0.708 & 0.706 & 0.695 & 0.684 \\
Fisher-regression (All) & 0.714 & 0.718 & 0.717 & 0.715 & 0.708 & 0.714 \\
Fisher-regression (PCA) & 0.726 & 0.74 & 0.725 & 0.729 & 0.721 & 0.715 \\
Fisher & 0.671 & 0.685 & 0.669 & 0.672 & 0.668 & 0.66 \\
\addlinespace
\addlinespace\multicolumn{7}{l}{\(K=30\)}\\
Approx. Fisher (PCA) & 0.692 & 0.711 & 0.698 & 0.697 & 0.677 & 0.678 \\
Fisher-regression (All) & 0.481 & 0.437 & 0.466 & 0.492 & 0.498 & 0.514 \\
Fisher-regression (PCA) & 0.699 & 0.709 & 0.706 & 0.707 & 0.686 & 0.688 \\
Fisher & 0.671 & 0.682 & 0.676 & 0.675 & 0.66 & 0.663 \\
\addlinespace
\addlinespace & \multicolumn{6}{c}{Correlated covariates} \\
\cmidrule(lr){2-7}
\multicolumn{7}{l}{\(K=10\)}\\
Approx. Fisher (PCA) & 0.772 & 0.811 & 0.789 & 0.769 & 0.757 & 0.733 \\
Fisher-regression (All) & 0.851 & 0.859 & 0.855 & 0.854 & 0.848 & 0.84 \\
Fisher-regression (PCA) & 0.86 & 0.871 & 0.865 & 0.862 & 0.855 & 0.848 \\
Fisher & 0.678 & 0.698 & 0.684 & 0.677 & 0.673 & 0.66 \\
\addlinespace
\addlinespace\multicolumn{7}{l}{\(K=20\)}\\
Approx. Fisher (PCA) & 0.737 & 0.766 & 0.746 & 0.736 & 0.727 & 0.711 \\
Fisher-regression (All) & 0.715 & 0.717 & 0.719 & 0.718 & 0.713 & 0.709 \\
Fisher-regression (PCA) & 0.787 & 0.8 & 0.794 & 0.787 & 0.781 & 0.775 \\
Fisher & 0.678 & 0.694 & 0.681 & 0.677 & 0.672 & 0.668 \\
\addlinespace
\addlinespace\multicolumn{7}{l}{\(K=30\)}\\
Approx. Fisher (PCA) & 0.719 & 0.736 & 0.723 & 0.718 & 0.718 & 0.701 \\
Fisher-regression (All) & 0.487 & 0.442 & 0.478 & 0.493 & 0.508 & 0.517 \\
Fisher-regression (PCA) & 0.749 & 0.76 & 0.753 & 0.746 & 0.749 & 0.738 \\
Fisher & 0.679 & 0.687 & 0.681 & 0.678 & 0.678 & 0.669 \\
\addlinespace
\addlinespace\bottomrule
\end{tabular*}
    \begin{tablenotes}
	    \item[] \footnotesize Note: The table shows the power from of a test of the sharp null at five percent significance level, with $\tau=1$. The first column shows the unconditional power, whereas the next five shows the power for each quintile of the Mahalanobis distance, $M_{\bm\Delta}$. For each value of $K$, 1,000 samples are drawn with 100 assignment vectors selected for each sample. The sample size is set to 50 and the following parameters are used: $n_s=1000$, $n_f=20$ and $\bar{\delta}=0.01$.
	\end{tablenotes}
\end{threeparttable}
\end{table}

Conditional on the quintiles of the Mahalanobis distance, we generally see the same pattern as in Table \ref{tab:power_homo}: the power is decreasing in the Mahalanobis distance. Our approximate Fisher-test has a steeper drop-off in power as the distance increases, likely because its size-distortion is smaller for large distances (c.f. Table \ref{tab:size_homo_fisher}; power is not size-adjusted).

Overall, we conclude that while the approximate Fisher-test may not be quite as powerful as regression-based Fisher-tests, it is closer to maintaining correct size conditional on observed imbalances in the covariates. In addition, with a large set of covariates, the suggested algorithm of reducing the number of covariates with the use of principal components substantially increases power of either the approximate Fisher-test or the principal component regression-based Fisher-test compared to the standard regression-based Fisher-test.

\section{Heterogeneous treatment effects}
We now turn to the study of heterogeneous treatment effects. To do so, we consider the following two linear projections:
\begin{align}
	Y_i(0) &= \alpha_0 + \mathbf{z}_i'\bm\upbeta_0 + \varepsilon_{0i} \\
	Y_i(1) &= \alpha_1 + \mathbf{z}_i'\bm\upbeta_1 + \varepsilon_{1i},
\end{align}
with the estimand of interest--the sample average treatment effect--being
\begin{equation}
	\tau = \frac{1}{n}\sum_{i=1}^n (Y_i(1)-Y_i(0)) = \alpha_1-\alpha_0 + \mathbf{\overline{z}}'(\bm\upbeta_1-\bm\upbeta_0),
\end{equation}
as $\overline\varepsilon_{0}=\overline\varepsilon_{1}=0$ by construction. By demeaning the linear projections and interacting with the treatment indicator, $W_i$, we can write the observed outcome as
\begin{equation}
	Y_i = \alpha_0^* + (\mathbf{z}_i-\overline{\mathbf{z}})'\bm\upbeta_0 + W_i\tau + W_i(\mathbf{z}_i-\overline{\mathbf{z}})'\bm\uprho + \eta_i, \label{eq:y_interact}
\end{equation}
where $\alpha_0^*=\alpha_0+\mathbf{\overline{z}}'\bm\upbeta_0$, $\eta_i=\varepsilon_{0i}+W_i( \varepsilon_{1i}- \varepsilon_{0i})$ and $\bm\uprho=(\bm\upbeta_1-\bm\upbeta_0)$. The difference-in-means estimator can be written as
\begin{equation}
 	\widehat{\tau }_{DM}=\overline{Y}_{1}-\overline{Y}_{0}=\tau+\bm\Delta^{\prime }\bm{\zeta}+\overline{%
 \varepsilon }_{11}-\overline{\varepsilon }_{00},
\end{equation}
where $\bm{\zeta} = \frac{n_1}{n}\bm\upbeta_0 + \frac{n_0}{n}\bm\upbeta_1$ and $\overline{ \varepsilon }_{11}$ and $\overline{ \varepsilon }_{00}$ are the respective averages of $\varepsilon_1$ and $\varepsilon_{0}$ in the treatment and control groups. Analogous to the case with homogeneous treatment effects, we have
\begin{equation}
E_{\mathcal{W}_{\bm{\Delta}} }(\widehat{\tau }_{DM}) =\tau +\bm{\Delta}^{\prime }\bm{\zeta} +E_{\mathcal{W}_{\bm{\Delta}} }(\overline{%
 \varepsilon }_{11}-\overline{\varepsilon }_{00}).
\end{equation}%
Once again, we have conditional bias in the difference-in-means estimator for $\bm\Delta\neq \mathbf{0}$. Naturally, the conditional variance of the difference-in-means estimator is
\begin{equation}
V_{\mathcal{W}_{\bm{\Delta}} }(\widehat{\tau }_{DM}) =V_{\mathcal{W}_{\bm{\Delta}} }(\overline{
 \varepsilon }_{11}-\overline{\varepsilon }_{00})
\end{equation}%

When it comes to the OLS estimator, equation \eqref{eq:y_interact} suggests that to properly deal with the case of heterogeneous treatment effects, all covariates should be demeaned and included both by themselves as well as interacted with the treatment indicator. The coefficient in front of the treatment indicator by itself is then an estimator for $\tau$. We can include all the control variables, including interactions in a $n\times 2K$ matrix $\mathbf{X}=\left[\begin{array}{cc}
	\widetilde{\mathbf{Z}} & \mathbf{Q}
\end{array}\right]$ with the $i$th row equaling $\mathbf{x}_i=\left[\begin{array}{cc}\mathbf{z}_i-\overline{\mathbf{z}} & (\mathbf{z}_i-\overline{\mathbf{z}}) W_i \end{array}\right]$. Let $\widetilde{\mathbf{X}}=\left[\begin{array}{cc}
	\widetilde{\mathbf{Z}} & \widetilde{\mathbf{Q}}
\end{array}\right]=\mathbf{X}-\mathbf{\overline{x}}$, $\widetilde{\bm\upeta} = \bm\upeta - \overline{\bm\upeta}$, $\widetilde{\mathbf{M}}_{x}=\mathbf{I}-\widetilde{%
\mathbf{X}}\mathbf{(\widetilde{\mathbf{X}}}^{\prime }\widetilde{\mathbf{X}}%
\mathbf{)}^{-1}\widetilde{\mathbf{X}}^{\prime }$ and $\bm\Theta = \left[\begin{array}{cc}\bm\upbeta_0' & \bm\uprho' \end{array}\right]'$, the OLS estimator of $\tau$, $\hat\tau_x$, can then be written as
\begin{equation}
	\widehat{\tau }_{x}=\mathbf{(\widetilde{\mathbf{W}}^{\prime }\widetilde{%
	\mathbf{M}}}_x\widetilde{\mathbf{W}}\mathbf{)}^{-1}\widetilde{\mathbf{W}}%
	^{\prime }\widetilde{\mathbf{M}}_x \widetilde{\mathbf{Y}},
\end{equation}
where
\begin{equation}
	\widetilde{\mathbf{Y}} = \widetilde{\mathbf{X}}\bm\Theta + \widetilde{\mathbf{W}} \tau + \widetilde{\bm\upeta}.
\end{equation}
Because $\widetilde{\mathbf{W}}^{\prime }\widetilde{\mathbf{M}}_x\widetilde{\mathbf{X}}\bm\Theta = 0$ and $\widetilde{\mathbf{W}}%
	^{\prime }\widetilde{\mathbf{M}}_x  \widetilde{\bm\upeta} = \widetilde{\mathbf{W}}%
	^{\prime }\widetilde{\mathbf{M}}_x  \bm\upeta$, we have
\begin{equation}
	\widehat{\tau }_{x}=\tau + \mathbf{(\widetilde{\mathbf{W}}^{\prime }\widetilde{%
	\mathbf{M}}}_x\widetilde{\mathbf{W}}\mathbf{)}^{-1}\widetilde{\mathbf{W}}%
	^{\prime }\widetilde{\mathbf{M}}_x \bm{\upeta}.
\end{equation}
The numerator equals
\begin{align}
	\widetilde{\mathbf{W}}^{\prime }\widetilde{\mathbf{M}}_x \bm\upeta &= \widetilde{\mathbf{W}}^{\prime }\bm\upeta-\widetilde{\mathbf{W}}^{\prime }\widetilde{%
\mathbf{X}}\mathbf{(\widetilde{\mathbf{X}}}^{\prime }\widetilde{\mathbf{X}}%
\mathbf{)}^{-1}\widetilde{\mathbf{X}}^{\prime }\bm\upeta,
\end{align}
where
\begin{equation}
	\widetilde{\mathbf{W}}^{\prime }\bm\upeta = \frac{n_0n_1}{n}(\overline{\varepsilon}_{11}-\overline{\varepsilon}_{00}).
\end{equation}
Furthermore,
\begin{equation}
	\widetilde{\mathbf{W}}^{\prime }\widetilde{\mathbf{X}} =
	\left[
		\begin{array}{cc}
			\frac{n_0n_1}{n}\bm\Delta' & \frac{n_0^2n_1}{n^2}\bm\Delta'
		\end{array}
	\right],
\end{equation}
and
\begin{equation}
	\widetilde{\mathbf{X}}^{\prime }\bm\upeta =
	\left[
		\begin{array}{cc}
			\mathbf{T}_1' & \mathbf{T}_2'
		\end{array}
	\right]',
\end{equation}
where $\mathbf{T}_1 = \left[ \begin{array}{c}\sum_{i=1}^n \widetilde{\mathbf{z}}_{i}\eta_i \end{array}\right]'$ and $\mathbf{T}_2=\left[ \begin{array}{c}\sum_{i:W_i=1}^n \widetilde{\mathbf{z}}_{i}\varepsilon_{i1} - \frac{n_0n_1}{n} \overline{\eta} \bm\Delta'\end{array}\right]'$. Finally, $\left(  \widetilde{\mathbf{X}}'\widetilde{\mathbf{X}}  \right)^{-1}$ can be partitioned as
\begin{equation}
	(\widetilde{\mathbf{X}}'\widetilde{\mathbf{X}})^{-1} = \frac{1}{n-1}\left[
		\begin{array}{cc}
			\bm\Sigma_{11}^{-1}& \bm\Sigma_{12}^{-1} \\
			\bm\Sigma_{12}'^{-1} & \bm\Sigma_{22}^{-1}
		\end{array}
		\right]
\end{equation} 
We get
\begin{align}
	\widetilde{\mathbf{W}}^{\prime }\widetilde{%
	\mathbf{X}}\mathbf{(\widetilde{\mathbf{X}}}^{\prime }\widetilde{\mathbf{X}}%
	\mathbf{)}^{-1}\widetilde{\mathbf{X}}^{\prime }\bm\upeta &= \frac{1}{n-1}
	\left[
		\begin{array}{cc}
			\frac{n_0n_1}{n}\bm\Delta' & \frac{n_0^2n_1}{n^2}\bm\Delta'
		\end{array}
	\right]
	\left[
		\begin{array}{cc}
			\bm\Sigma_{11}^{-1}& \bm\Sigma_{12}^{-1} \\
			\bm\Sigma_{12}'^{-1} & \bm\Sigma_{22}^{-1}
		\end{array}
	\right]
	\left[
		\begin{array}{c}
			\mathbf{T}_1 \\
			\mathbf{T}_2
		\end{array}
	\right] \nonumber \\
	&= \frac{n_0n_1}{n(n-1)}\left(\left(\bm\Delta' \bm\Sigma_{11}^{-1} + \frac{n_0}{n}\bm\Delta' \bm\Sigma_{12}'^{-1}\right)\mathbf{T}_1 +
	\left(\bm\Delta' \bm\Sigma_{12}^{-1} + \frac{n_0}{n}\bm\Delta' \bm\Sigma_{22}^{-1}\right)\mathbf{T}_2\right).
\end{align}
The denominator of the OLS estimator:
\begin{align}
	\widetilde{\mathbf{W}}^{\prime }\widetilde{\mathbf{M}}_x\widetilde{\mathbf{W}}&=\mathbf{\widetilde{\mathbf{W}}}^{\prime }%
\mathbf{\widetilde{\mathbf{W}}}-\mathbf{\widetilde{\mathbf{W}}}^{\prime }%
\widetilde{\mathbf{X}}\mathbf{(\widetilde{\mathbf{X}}}^{\prime }%
\widetilde{\mathbf{X}}\mathbf{)}^{-1}\widetilde{\mathbf{X}}^{\prime }%
\mathbf{\widetilde{\mathbf{W}}} \nonumber \\
&=
\frac{n_0n_1}{n}\left(1-\frac{n_0n_1}{n(n-1)}\left(\bm\Delta' \bm\Sigma_{11}^{-1}\bm\Delta + \frac{n_0^2}{n^2}  \bm\Delta' \bm\Sigma_{22}^{-1}\bm\Delta + 2\frac{n_0}{n}  \bm\Delta' \bm\Sigma_{12}^{-1}\bm\Delta\right) \right).
\end{align}
The OLS estimator can the be written as
\begin{equation}
	\widehat{\tau }_{x} = \tau + \frac{\overline{\varepsilon}_{11}-\overline{\varepsilon}_{00} - \frac{1}{n-1}\left(\left(\bm\Delta' \bm\Sigma_{11}^{-1} + \frac{n_0}{n}\bm\Delta' \bm\Sigma_{12}'^{-1}\right)\mathbf{T}_1 +
	\left(\bm\Delta' \bm\Sigma_{12}^{-1} + \frac{n_0}{n}\bm\Delta' \bm\Sigma_{22}^{-1}\right)\mathbf{T}_2\right)}{1-\frac{n_0n_1}{n(n-1)}\left(\bm\Delta' \bm\Sigma_{11}^{-1}\bm\Delta + \frac{n_0^2}{n^2}  \bm\Delta' \bm\Sigma_{22}^{-1}\bm\Delta + 2\frac{n_0}{n}  \bm\Delta' \bm\Sigma_{12}^{-1}\bm\Delta\right)}.
\end{equation}
The difference from the case with homogeneous treatment effects is that the conditional bias of the OLS estimator no longer depends solely on $\overline{\varepsilon}_{11}-\overline{\varepsilon}_{00}$, but also on $\mathbf{T}_1$ and $\mathbf{T}_2$; in the homogeneous case, $\mathbf{Z}'\bm\upvarepsilon = \mathbf{0}$, whereas in the heterogeneous case, $\mathbf{Z}'\bm\upeta \neq \mathbf{0}$. The only time this second term disappears is when $\bm\Delta=\mathbf{0}$.

We can write the denominator in the expression above as $1 - \widetilde{M}_{\bm\Delta}/(n-1)$, where $\widetilde{M}_{\bm\Delta}$ is a weighted Mahalanobis distance of $\bm\Delta$ (technically, it is the Mahalanobis distance of $\overline{\mathbf{x}}_1-\overline{\mathbf{x}}_0$). The expectation of the OLS estimator is
\begin{multline}
	E_{\mathcal{W}_{\bm{\Delta}} }(\widehat{\tau }_x) = \tau + 
	\frac{E_{\mathcal{W}_{\bm{\Delta}} }(\overline{\varepsilon }_{11}-\overline{\varepsilon }_{00})}{1 - \widetilde{M}_{\bm\Delta}/(n-1)} - \\ \frac{\frac{1}{n-1}\left(\left(\bm\Delta' \bm\Sigma_{11}^{-1} + \frac{n_0}{n}\bm\Delta' \bm\Sigma_{12}'^{-1}\right)E_{\mathcal{W}_{\bm{\Delta}} }(\mathbf{T}_1) +
	\left(\bm\Delta' \bm\Sigma_{12}^{-1} + \frac{n_0}{n}\bm\Delta' \bm\Sigma_{22}^{-1}\right)E_{\mathcal{W}_{\bm{\Delta}} }(\mathbf{T}_2)\right)}{1 - \widetilde{M}_{\bm\Delta}/(n-1)}.
\end{multline}

\subsection{Regression-based inference}

\cite{freedman_regression_2008} discusses the extent to which randomization justifies regression adjustment in the Neyman model \citep{splawa-neyman_application_1990} and studies the asymptotic properties when the number of units in the experiment goes to infinity. He shows (i) that the OLS covariate adjustment estimator is, in general, biased (of order $1/n$), (ii) that the conventional OLS\ estimated standard errors estimator is inconsistent, and (iii) that, with unbalanced designs, the OLS estimator also could be less efficient than the difference-in-means estimator asymptotically.  However, \cite{lin_agnostic_2013} shows (i) that the Eicker-Huber-White standard error estimator \citep{Eicker_1967,Huber_1967, White_1980} is consistent or asymptotically conservative and (ii) that the OLS estimator from equation \eqref{eq:y_interact} is, asymptotically, at least as efficient as the difference-in-means estimator. This is the procedure for inference we use for the regression estimator in the simulations below. 

\subsubsection{Simulations}
To study the properties of the estimators with heterogeneous treatment effects in a simulation study, we generate data as
\begin{align}
	Y_i(0) &= \mathbf{z}_i\mathbf{b} + u_{0i}, \label{eq:dgp_y0} \\
	Y_i(1) &= \mathbf{z}_i\mathbf{b} + \gamma + u_{1i}. \label{eq:dgp_y1}
\end{align}
where $\mathbf{Z}\sim N(\mathbf{0}, \mathbf{I})$, $\mathbf{b}=\left[\begin{array}{cccc}
	\frac{1}{\sqrt{K}} & \frac{1}{\sqrt{K}} & \ldots & \frac{1}{\sqrt{K}} \\
\end{array}\right]$ and both $u_0$ and $u_1$ following standard normal distributions. In these simulations, the heterogeneity therefore comes solely from the differing errors. As before, for our algorithm, we let $\bar{\delta}=0.01$ and the sample size is set to $n=50$ with $n_0=n_1=25$. Because we can use up to $2K$ covariates in a regression (because of the interactions), we vary $K$ from 2 to 20 in steps of 1. $n_{\bar{\delta}}$ in Algorithm \ref{alg:ComponentSelection} is now initiated at $n_A/2$.

Figure \ref{fig:sim_results_2_20_hetero} shows the results for the MSE (calculated for $\tau$, the sample average treatment effect and not $\gamma$, the population average treatment effect). The result is very similar to that in Figure \ref{fig:sim_results_2_40} with the difference that our PCA estimator always perform as good or better than the OLS estimator.

\begin{figure}
	\captionsetup{position=top}
	\caption{MSE, heterogeneous treatment effect \label{fig:sim_results_2_20_hetero}}
	\includegraphics[width=\linewidth]{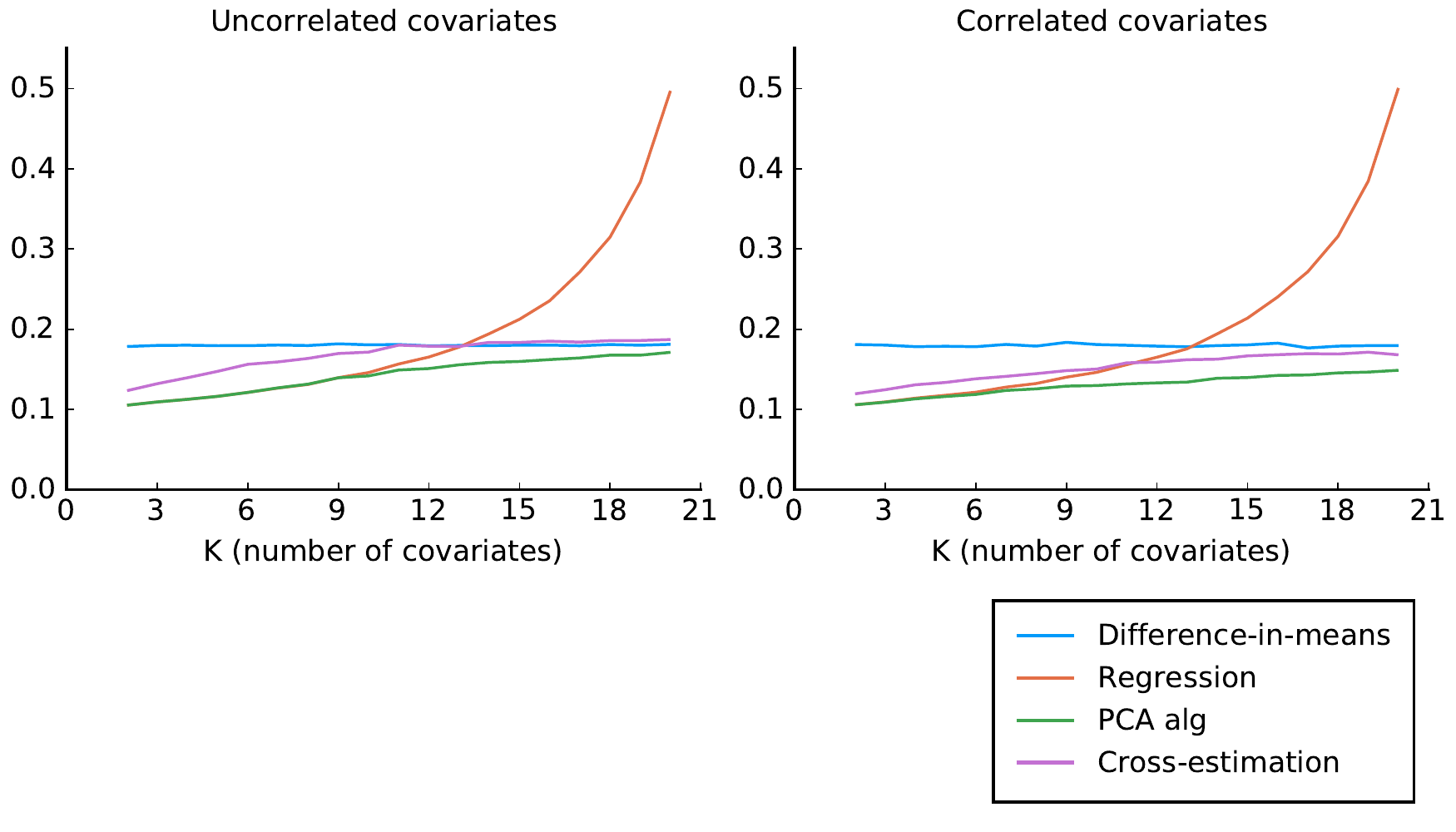}
	\floatfoot{\footnotesize Note: For each value of $K$, 1,000 samples are drawn with 10,000 assignment vectors selected for each sample. For the cross-estimation estimator, for computational time purposes, only 100 assignment vectors are selected for each sample. The sample size is set to 50 and $\gamma=0$.}
\end{figure}

Figure \ref{fig:mse_cond_M_hetero} shows the MSE conditional on the Mahalanobis distance in the same way as in Figure \ref{fig:mse_cond_M}. The OLS and PCA estimators are virtually identical when $K=5$, with the PCA estimator outperforming the OLS estimator for larger values of $K$. Overall, the conclusions are similar when treatment effects are heterogeneous as compared to when they are homogeneous.

\begin{figure}
	\captionsetup{position=top}
	\caption{MSE by percentile of the original Mahalanobis distance \label{fig:mse_cond_M_hetero}}
	\subfloat[Uncorrelated covariates]{
		\includegraphics[width=\linewidth]{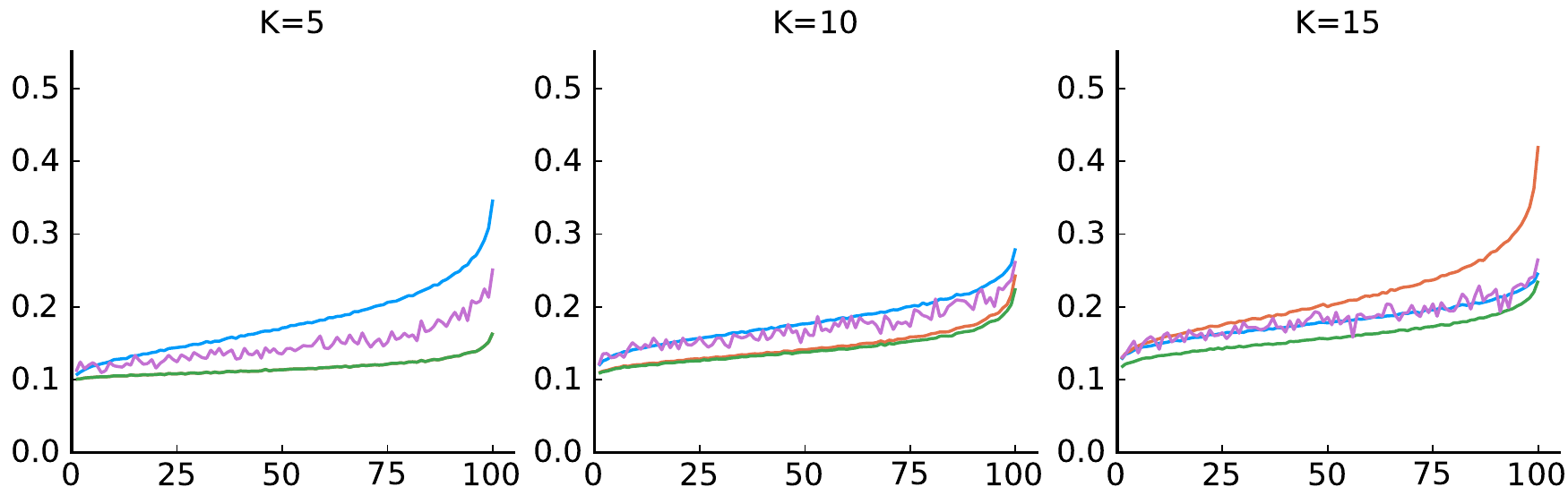}}\\
	\subfloat[Correlated covariates]{
		\includegraphics[width=\linewidth]{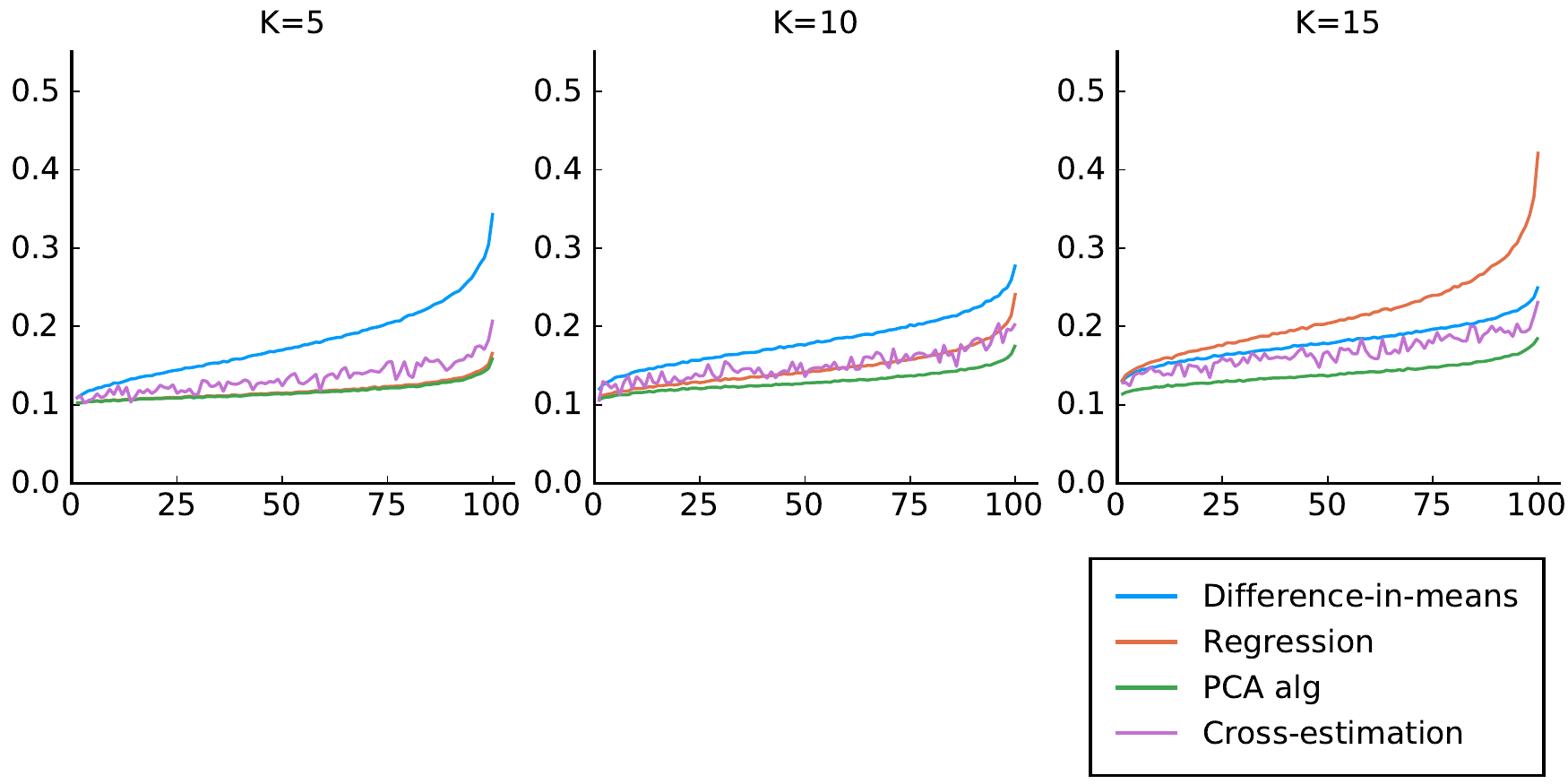}}
	\floatfoot{\footnotesize Note: For each value of $K$, 1,000 samples are drawn with 10,000 assignment vectors selected for each sample. For the cross-estimation estimator, for computational time purposes, only 100 assignment vectors are selected for each sample. The sample size is set to 50 and $\gamma=0$.}
\end{figure}

Table \ref{tab:size_hetero_hetero} shows the size of a test where the null is set to the sample average treatment effect at five percent significance level. We now see that no estimator gives correct size, with the difference-in-means, OLS and PCA estimators all typically being conservative, while the cross-estimation estimator continuous to overreject. However, the average size-distortion is in general quite small.

Conditional on the Mahalanobis distance, the same pattern as before is present: as the Mahalanobis distance increases, the rejection rate for all estimators increase. Different from the case with homogeneous treatment effects, this is true also for the OLS estimator.

\begin{table}[p!]
	\begin{threeparttable}
	    \caption{Size, heterogeneous effects \label{tab:size_hetero_hetero}}
	    \begin{tabular*}{\textwidth}{l @{\extracolsep{\fill}} cccccc}
\toprule
Quintiles & All & 1st & 2nd & 3rd & 4th & 5th \\
\midrule & \multicolumn{6}{c}{Uncorrelated covariates} \\
\cmidrule(lr){2-7}
\multicolumn{7}{l}{\(K=5\)}\\
Difference-in-means & 0.039 & 0.014 & 0.023 & 0.033 & 0.047 & 0.08 \\
Cross-estimation & 0.058 & 0.038 & 0.046 & 0.053 & 0.063 & 0.088 \\
Regression & 0.038 & 0.035 & 0.036 & 0.037 & 0.039 & 0.042 \\
PCA alg & 0.038 & 0.035 & 0.036 & 0.037 & 0.039 & 0.042 \\
\addlinespace
\addlinespace\multicolumn{7}{l}{\(K=10\)}\\
Difference-in-means & 0.04 & 0.02 & 0.029 & 0.037 & 0.046 & 0.066 \\
Cross-estimation & 0.063 & 0.044 & 0.051 & 0.06 & 0.068 & 0.091 \\
Regression & 0.044 & 0.04 & 0.042 & 0.044 & 0.046 & 0.05 \\
PCA alg & 0.043 & 0.036 & 0.039 & 0.041 & 0.045 & 0.053 \\
\addlinespace
\addlinespace\multicolumn{7}{l}{\(K=15\)}\\
Difference-in-means & 0.039 & 0.023 & 0.032 & 0.038 & 0.045 & 0.059 \\
Cross-estimation & 0.068 & 0.047 & 0.059 & 0.066 & 0.076 & 0.093 \\
Regression & 0.052 & 0.047 & 0.05 & 0.052 & 0.054 & 0.058 \\
PCA alg & 0.044 & 0.034 & 0.039 & 0.043 & 0.047 & 0.057 \\
\addlinespace
\addlinespace & \multicolumn{6}{c}{Correlated covariates} \\
\cmidrule(lr){2-7}
\multicolumn{7}{l}{\(K=5\)}\\
Difference-in-means & 0.039 & 0.016 & 0.023 & 0.033 & 0.046 & 0.077 \\
Cross-estimation & 0.054 & 0.039 & 0.044 & 0.051 & 0.061 & 0.075 \\
Regression & 0.038 & 0.035 & 0.036 & 0.037 & 0.039 & 0.042 \\
PCA alg & 0.037 & 0.034 & 0.036 & 0.037 & 0.039 & 0.041 \\
\addlinespace
\addlinespace\multicolumn{7}{l}{\(K=10\)}\\
Difference-in-means & 0.039 & 0.02 & 0.029 & 0.036 & 0.046 & 0.065 \\
Cross-estimation & 0.059 & 0.044 & 0.051 & 0.053 & 0.065 & 0.08 \\
Regression & 0.044 & 0.039 & 0.042 & 0.044 & 0.046 & 0.05 \\
PCA alg & 0.041 & 0.037 & 0.04 & 0.041 & 0.042 & 0.046 \\
\addlinespace
\addlinespace\multicolumn{7}{l}{\(K=15\)}\\
Difference-in-means & 0.039 & 0.024 & 0.031 & 0.038 & 0.045 & 0.058 \\
Cross-estimation & 0.066 & 0.047 & 0.057 & 0.063 & 0.073 & 0.088 \\
Regression & 0.052 & 0.046 & 0.049 & 0.052 & 0.054 & 0.057 \\
PCA alg & 0.042 & 0.037 & 0.04 & 0.041 & 0.044 & 0.049 \\
\addlinespace
\addlinespace\bottomrule
\end{tabular*}
	    \begin{tablenotes}
	    	\item[] \footnotesize Note: The table shows the size of a test where the null is the sample average treatment effect at five percent significance level. The first column shows the unconditional size, whereas the next five shows the size for each quintile of the Mahalanobis distance, $M_{\bm\Delta}$. For each value of $K$, 1,000 samples are drawn with 10,000 assignment vectors selected for each sample. For the cross-estimation estimator, for computational time purposes, only 100 assignment vectors are selected for each sample. The sample size is set to 50. For the regression-based estimators, the Eicker-Huber-White robust covariance matrix is used.
	    \end{tablenotes}
	\end{threeparttable}
\end{table}

Finally, Table \ref{tab:power_hetero} shows the result corresponding to Table \ref{tab:power_homo} in the homogeneous case. Because we are interested in studying the power, the null is now set to zero instead of the sample average treatment effect. Consistent with the results shown in Figure \ref{fig:sim_results_2_20_hetero}, the PCA estimator generally outperforms the other three estimators on average, as well as conditionally for small Mahalanobis distances. The cross-estimation estimator is generally slightly more powerful for large distances and roughly equally powerful for $K=15$. However, this can partly be attributed to the fact that the test rejects the null slightly too often.

\begin{table}[p!]
	\begin{threeparttable}
	    \caption{Power, heterogeneous effects \label{tab:power_hetero}}
	    \begin{tabular*}{\textwidth}{l @{\extracolsep{\fill}} cccccc}
\toprule
Quintiles & All & 1st & 2nd & 3rd & 4th & 5th \\
\midrule & \multicolumn{6}{c}{Uncorrelated covariates} \\
\cmidrule(lr){2-7}
\multicolumn{7}{l}{\(K=5\)}\\
Difference-in-means & 0.588 & 0.599 & 0.594 & 0.589 & 0.583 & 0.575 \\
Cross-estimation & 0.711 & 0.722 & 0.718 & 0.716 & 0.707 & 0.695 \\
Regression & 0.748 & 0.782 & 0.768 & 0.754 & 0.737 & 0.698 \\
PCA alg & 0.748 & 0.782 & 0.768 & 0.754 & 0.737 & 0.698 \\
\addlinespace
\addlinespace\multicolumn{7}{l}{\(K=10\)}\\
Difference-in-means & 0.589 & 0.597 & 0.592 & 0.589 & 0.585 & 0.58 \\
Cross-estimation & 0.674 & 0.683 & 0.681 & 0.68 & 0.666 & 0.659 \\
Regression & 0.667 & 0.733 & 0.7 & 0.674 & 0.644 & 0.587 \\
PCA alg & 0.68 & 0.729 & 0.702 & 0.682 & 0.661 & 0.625 \\
\addlinespace
\addlinespace\multicolumn{7}{l}{\(K=15\)}\\
Difference-in-means & 0.588 & 0.594 & 0.59 & 0.588 & 0.585 & 0.582 \\
Cross-estimation & 0.65 & 0.66 & 0.652 & 0.647 & 0.65 & 0.639 \\
Regression & 0.537 & 0.631 & 0.577 & 0.541 & 0.501 & 0.435 \\
PCA alg & 0.64 & 0.678 & 0.655 & 0.64 & 0.624 & 0.601 \\
\addlinespace
\addlinespace & \multicolumn{6}{c}{Correlated covariates} \\
\cmidrule(lr){2-7}
\multicolumn{7}{l}{\(K=5\)}\\
Difference-in-means & 0.592 & 0.6 & 0.596 & 0.593 & 0.589 & 0.583 \\
Cross-estimation & 0.742 & 0.755 & 0.75 & 0.745 & 0.738 & 0.723 \\
Regression & 0.753 & 0.787 & 0.772 & 0.76 & 0.742 & 0.703 \\
PCA alg & 0.757 & 0.788 & 0.774 & 0.763 & 0.747 & 0.712 \\
\addlinespace
\addlinespace\multicolumn{7}{l}{\(K=10\)}\\
Difference-in-means & 0.601 & 0.608 & 0.604 & 0.602 & 0.599 & 0.594 \\
Cross-estimation & 0.703 & 0.716 & 0.708 & 0.705 & 0.698 & 0.688 \\
Regression & 0.673 & 0.737 & 0.706 & 0.68 & 0.65 & 0.593 \\
PCA alg & 0.716 & 0.754 & 0.734 & 0.719 & 0.702 & 0.672 \\
\addlinespace
\addlinespace\multicolumn{7}{l}{\(K=15\)}\\
Difference-in-means & 0.597 & 0.603 & 0.6 & 0.597 & 0.595 & 0.591 \\
Cross-estimation & 0.683 & 0.691 & 0.69 & 0.68 & 0.677 & 0.675 \\
Regression & 0.537 & 0.632 & 0.578 & 0.541 & 0.501 & 0.434 \\
PCA alg & 0.691 & 0.726 & 0.706 & 0.692 & 0.678 & 0.654 \\
\addlinespace
\addlinespace\bottomrule
\end{tabular*}
	    \begin{tablenotes}
	    	\item[] \footnotesize Note: The table shows the power from of a test of $\tau=0$ at five percent significance level, with $\gamma=1$. The first column shows the unconditional power, whereas the next five shows the power for each quintile of the Mahalanobis distance, $M_{\bm\Delta}$. For each value of $K$, 1,000 samples are drawn with 10,000 assignment vectors selected for each sample. For the cross-estimation estimator, for computational time purposes, only 100 assignment vectors are selected for each sample. The sample size is set to 50. For the regression-based estimators, the Eicker-Huber-White robust covariance matrix is used.
	    \end{tablenotes}
	\end{threeparttable}
\end{table}

Overall, the conclusions from the simulations on homogeneous treatment effects carry over to the heterogeneous case. We find that the PCA estimator generally performs the best by having the smallest MSE and highest power, while being slightly conservative in terms of test size.

\subsection{Randomization-based inference, simulations}
To study inference with the Fisher tests when treatment effects are heterogeneous, we only study power, as the sharp null will always be false with heterogeneous effects. Data is generated according to equations \eqref{eq:dgp_y0} and \eqref{eq:dgp_y1} with $\gamma=1$. Results are shown in Table \ref{tab:power_hetero_fisher}. Once again, results are very similar to the case with homogeneous treatment effects (Table \ref{tab:power_homo_fisher}): The regression-based tests are in general the most powerful (except whan all fifteen covariates are used) followed by the approximate Fisher-test. For small Mahalanobis-distances, the differences between the approximate test and the regression-based tests are in general small, but as the distance gets larger, the differences increase, something that, at least partially, can be explained by the relatively larger size-distortion in the conditional regression-based tests (see Table \ref{tab:size_homo_fisher}).

\begin{table}[p!]
\begin{threeparttable}
    \caption{Power, heterogeneous effects, Fisher tests \label{tab:power_hetero_fisher}}
    \begin{tabular*}{\textwidth}{l @{\extracolsep{\fill}} cccccc}
\toprule
Quintiles & All & 1st & 2nd & 3rd & 4th & 5th \\
\midrule & \multicolumn{6}{c}{Uncorrelated covariates} \\
\cmidrule(lr){2-7}
\multicolumn{7}{l}{\(K=5\)}\\
Approx. Fisher (PCA) & 0.666 & 0.724 & 0.684 & 0.665 & 0.644 & 0.615 \\
Fisher-regression (All) & 0.744 & 0.75 & 0.744 & 0.747 & 0.742 & 0.738 \\
Fisher-regression (PCA) & 0.744 & 0.749 & 0.746 & 0.745 & 0.741 & 0.736 \\
Fisher & 0.572 & 0.584 & 0.578 & 0.574 & 0.564 & 0.561 \\
\addlinespace
\addlinespace\multicolumn{7}{l}{\(K=10\)}\\
Approx. Fisher (PCA) & 0.635 & 0.669 & 0.648 & 0.637 & 0.622 & 0.599 \\
Fisher-regression (All) & 0.701 & 0.708 & 0.701 & 0.699 & 0.698 & 0.697 \\
Fisher-regression (PCA) & 0.669 & 0.679 & 0.67 & 0.668 & 0.668 & 0.659 \\
Fisher & 0.586 & 0.592 & 0.59 & 0.588 & 0.582 & 0.577 \\
\addlinespace
\addlinespace\multicolumn{7}{l}{\(K=15\)}\\
Approx. Fisher (PCA) & 0.608 & 0.629 & 0.614 & 0.605 & 0.605 & 0.584 \\
Fisher-regression (All) & 0.614 & 0.614 & 0.611 & 0.62 & 0.611 & 0.615 \\
Fisher-regression (PCA) & 0.625 & 0.629 & 0.627 & 0.623 & 0.626 & 0.619 \\
Fisher & 0.573 & 0.581 & 0.569 & 0.573 & 0.574 & 0.568 \\
\addlinespace
\addlinespace & \multicolumn{6}{c}{Correlated covariates} \\
\cmidrule(lr){2-7}
\multicolumn{7}{l}{\(K=5\)}\\
Approx. Fisher (PCA) & 0.663 & 0.718 & 0.68 & 0.66 & 0.641 & 0.618 \\
Fisher-regression (All) & 0.736 & 0.742 & 0.738 & 0.736 & 0.731 & 0.731 \\
Fisher-regression (PCA) & 0.739 & 0.746 & 0.739 & 0.739 & 0.736 & 0.734 \\
Fisher & 0.577 & 0.584 & 0.576 & 0.578 & 0.575 & 0.572 \\
\addlinespace
\addlinespace\multicolumn{7}{l}{\(K=10\)}\\
Approx. Fisher (PCA) & 0.654 & 0.688 & 0.67 & 0.648 & 0.642 & 0.622 \\
Fisher-regression (All) & 0.684 & 0.693 & 0.689 & 0.68 & 0.68 & 0.679 \\
Fisher-regression (PCA) & 0.709 & 0.716 & 0.712 & 0.706 & 0.707 & 0.705 \\
Fisher & 0.585 & 0.59 & 0.588 & 0.582 & 0.59 & 0.575 \\
\addlinespace
\addlinespace\multicolumn{7}{l}{\(K=15\)}\\
Approx. Fisher (PCA) & 0.634 & 0.666 & 0.645 & 0.626 & 0.623 & 0.611 \\
Fisher-regression (All) & 0.614 & 0.612 & 0.615 & 0.617 & 0.612 & 0.616 \\
Fisher-regression (PCA) & 0.674 & 0.679 & 0.673 & 0.673 & 0.673 & 0.672 \\
Fisher & 0.574 & 0.582 & 0.571 & 0.567 & 0.573 & 0.576 \\
\addlinespace
\addlinespace\bottomrule
\end{tabular*}
    \begin{tablenotes}
	    \item[] \footnotesize Note: The table shows the power from of a test of $\tau=0$ at five percent significance level, with $\gamma=1$. The first column shows the unconditional power, whereas the next five shows the power for each quintile of the Mahalanobis distance, $M_{\bm\Delta}$. For each value of $K$, 1,000 samples are drawn with 100 assignment vectors selected for each sample. The sample size is set to 50 and the following parameters are used: $n_s=1000$, $n_f=20$ and $\bar{\delta}=0.01$.
	\end{tablenotes}
\end{threeparttable}
\end{table}

\section{Concluding discussion}
Randomized controlled trials are considered the gold standard for causal inferences as randomization of treatment guarantees that the difference-in-means estimator is an unbiased estimator of the average treatment effect under no model assumption. However, this unbiasedness only holds under randomization over all possible assignment vectors.

Indeed, in this paper we show that conditional on observed imbalances in covariates, the difference-in-means estimator is in general biased, with associated statistical tests having incorrect size. As researchers are generally encouraged to investigate whether covariates are balanced, this fact puts the practitioner in an awkward position: on the one hand, the estimator is unbiased over all possible randomizations; on the other hand, conditional on the differences actually observed, the estimator is most likely biased.

A solution to this problem is to condition on observed covariates in a regression model, and we show that the OLS estimator is approximately conditionally unbiased. On the other hand, \cite{athey_econometrics_2017} cautions against the use of the OLS estimator in analyzing randomized experiments as the OLS estimator was not developed with randomization inference in mind, resulting in a disconnect between the assumptions needed for regression and for randomized controlled trials. Specifically, they write that ``it is easy for the researcher using regression methods to go beyond analyses that are justified by randomization, and end up with analyses that rely on a difficult-to-assess mix of randomization assumptions, modeling assumptions, and large sample approximations''. Similarly, \cite{freedman_regression_2008} writes that ``Regression adjustments are often made to experimental data. Since randomization does not justify the models, almost anything can happen''. 

Furthermore---and as discussed in \cite{mutz_perils_2019}---if practitioners adjust for covariates only when they are imbalanced between treatment and control groups, the inference  will be compromised. A further problem also discussed in \cite{mutz_perils_2019} is that with many covariates, many different regression estimators are possible raising the concern of ``p-hacking''. With these objections in mind---and with the need to avoid adding all covariates in the regression model to avoid a high MSE---we develop an algorithm based on the principal components of the covariates and select only so many principal components that can be justified based on randomization inference, thereby alleviating the concerns raised by \cite{athey_econometrics_2017}, \cite{freedman_regression_2008} and \cite{ mutz_perils_2019}.


In addition, we also develop a version of Fisher's exact test where, instead of comparing the treatment effect estimate from the assignment vector actually chosen with all other assignment vectors, we only select a small subset of assignment vectors very similar to the chosen assignment vector to perform the test. With this approach, the test will have approximately correct size even conditional on observed covariate imbalances.

\bibliography{ref}
\clearpage
\appendix
\section*{Appendix}
\subsection*{Explicit formulas with a single dummy variable as covariate}
To derive explicit formulas for $E_{\mathcal{W}_{\bm{\Delta}} }(\overline{\varepsilon }_{1}-\overline{\varepsilon }_{0})$ and $V_{\mathcal{W}_{\bm{\Delta}} }(\overline{\varepsilon }_{1}-\overline{\varepsilon }_{0})$, we consider the case with a single covariate, $\mathbf{z}$, that can take the value 0 or 1, where $n_Z=\sum_{i=1}^n Z_i$. Let $\mathcal{Z}$ be the set of indexes for when $Z=1$ and let $\mathcal{Z}^c=\{1,\ldots,n\}\setminus \mathcal{Z}$ be the set of indexes for when $Z=0$.

\begin{corollary}\label{cor:sym_freq}
	For $\mathcal{W}_{\bm\Delta} \subseteq \mathcal{W}$, $w,w'\in\{0,1\}$ and $z,z'\in\{0,1\}$, it is the case that 
	\begin{enumerate}[label=(\roman*)]
		\item over all assignment vectors in $\mathcal{W}_{\bm\Delta}$, the number of times treatment status $w$ occur is the same for all units $i$ with $Z_i = z$.
		\item over all assignment vectors in $\mathcal{W}_{\bm\Delta}$, the number of times treatment status $w,w'$ occur is the same for all units $i,j$ with $Z_i=z$ and $Z_j=z'$.
	\end{enumerate}
\end{corollary}
\begin{proof}
	An assignment vector with a given $\bm\Delta=\overline{z}_1-\overline{z}_0$ implies that there is a given number of treated and control units with $z=1$ and $z=0$. Because the set $\mathcal{W}_{\bm\Delta}$ contains all possible such assignment vectors, a unit $i$ with $Z_i=z$ has to have treatment status $w$ the same number of times over all assignment vectors in the set $\mathcal{W}_{\bm\Delta}$ as any other unit $i'$ with $Z_{i'}=z$. Similarly, two units $i,j$ with $Z_i=z$ and $Z_j=z'$ must have treatment status $w,w'$ the same number of times as any other two units $i',j'$ with $Z_{i'}=z$ and $Z_{j'}=z'$.
\end{proof}

By Corollary \ref{cor:sym_freq}(i), we can write the relative frequency at which unit $i$ is treated over all assignments in $\mathcal{W}_{\bm{\Delta}}$ with two different constants:
\begin{align}
	f_{1\bm{\Delta}}:=&|\{\mathbf{W} \in \mathcal{W}_{\bm{\Delta}} : W_i = 1 \}| / |\mathcal{W}_{\bm{\Delta}}|, \quad \forall i \in \mathcal{Z}, \\
	f_{2\bm{\Delta}}:=&|\{\mathbf{W} \in \mathcal{W}_{\bm{\Delta}} : W_i = 1 \}| / |\mathcal{W}_{\bm{\Delta}}|, \quad \forall i \in \mathcal{Z}^c.
\end{align}
We have
\begin{equation}
	E_{\mathcal{W}_{\bm{\Delta}} }(\overline{\varepsilon}_1 - \overline{\varepsilon}_0)=
	\frac{1}{n_1}\left( f_{1\bm{\Delta}} \sum_{i\in \mathcal{Z}} \varepsilon_i + f_{2\bm{\Delta}} \sum_{i\in \mathcal{Z}^c} \varepsilon_i \right) -
	\frac{1}{n_0}\left( (1-f_{1\bm{\Delta}}) \sum_{i\in \mathcal{Z}} \varepsilon_i + (1-f_{2\bm{\Delta}}) \sum_{i\in \mathcal{Z}^c} \varepsilon_i \right).
\end{equation}
Because $\mathbf{z}$ is included in the linear projection, it is the case that
\begin{equation}
 	\sum_{i\in \mathcal{Z}} \varepsilon_i=\sum_{i\in \mathcal{Z}^c} \varepsilon_i=0 \Rightarrow E_{\mathcal{W}_{\bm{\Delta}} }(\overline{\varepsilon}_1 - \overline{\varepsilon}_0)=0. \label{eq:sum_ez}
\end{equation}
We can therefore write the conditional variance as
\begin{equation}
V_{\mathcal{W}_{\bm{\Delta}} }(\overline{\varepsilon}_1 - \overline{\varepsilon}_0)=   E_{\mathcal{W}_{\bm{\Delta}} }\left((\overline{%
\varepsilon }_{1}-\overline{\varepsilon }_{0})^2\right) = E_{\mathcal{W}_{\bm{\Delta}} }\left(\overline{\varepsilon }_{1}^2\right) + E_{\mathcal{W}_{\bm{\Delta}} }\left(\overline{\varepsilon }_{0}^2\right) - 2E_{\mathcal{W}_{\bm{\Delta}} }\left(\overline{\varepsilon }_{1}\overline{\varepsilon }_{0}\right).
\end{equation}
By Corollary \ref{cor:sym_freq}(ii), we can define the following constants:
\begin{align}
	f_{3\bm{\Delta}}:=&|\{\mathbf{W} \in \mathcal{W}_{\bm{\Delta}} : W_i = 1 \land W_j = 1 \}| / |\mathcal{W}_{\bm{\Delta}}|, \quad \forall i,j \in \mathcal{Z}, i \neq j, \\
	f_{4\bm{\Delta}}:=&|\{\mathbf{W} \in \mathcal{W}_{\bm{\Delta}} : W_i = 1 \land W_j = 1 \}| / |\mathcal{W}_{\bm{\Delta}}|, \quad \forall i,j \in \mathcal{Z}^c, i \neq j, \\
	f_{5\bm{\Delta}}:=&|\{\mathbf{W} \in \mathcal{W}_{\bm{\Delta}} : W_i = 1 \land W_j = 1 \}| / |\mathcal{W}_{\bm{\Delta}}|, \quad \forall i \in \mathcal{Z}, \forall j \in \mathcal{Z}^c, \\
	f_{6\bm{\Delta}}:=&|\{\mathbf{W} \in \mathcal{W}_{\bm{\Delta}} : W_i = 1 \land W_j = 0 \}| / |\mathcal{W}_{\bm{\Delta}}|, \quad \forall i,j \in \mathcal{Z}, i \neq j, \\
	f_{7\bm{\Delta}}:=&|\{\mathbf{W} \in \mathcal{W}_{\bm{\Delta}} : W_i = 1 \land W_j = 0 \}| / |\mathcal{W}_{\bm{\Delta}}|, \quad \forall i,j \in \mathcal{Z}^c, i \neq j, \\
	f_{8\bm{\Delta}}:=&|\{\mathbf{W} \in \mathcal{W}_{\bm{\Delta}} : W_i = 1 \land W_j = 0 \}| / |\mathcal{W}_{\bm{\Delta}}|, \quad \forall i \in \mathcal{Z}, \forall j \in \mathcal{Z}^c, \\
	f_{9\bm{\Delta}}:=&|\{\mathbf{W} \in \mathcal{W}_{\bm{\Delta}} : W_i = 1 \land W_j = 0 \}| / |\mathcal{W}_{\bm{\Delta}}|, \quad \forall i \in \mathcal{Z}^c, \forall j \in \mathcal{Z}.
\end{align}
We can write $E_{\mathcal{W}_{\bm{\Delta}} }\left(\overline{\varepsilon }_{1}^2\right)$ as
\begin{multline}
	E_{\mathcal{W}_{\bm{\Delta}} }\left(\overline{\varepsilon }_{1}^2\right) =
	\frac{1}{n_1^2} \left(
		f_{1\bm{\Delta}} \sum_{i\in \mathcal{Z}} \varepsilon_i^2 +
		f_{2\bm{\Delta}} \sum_{i\in \mathcal{Z}^c} \varepsilon_i^2 +
		f_{3\bm{\Delta}} \sum_{i\in \mathcal{Z}}\sum_{j\in \mathcal{Z}\setminus\{i\}} \varepsilon_i\varepsilon_j + \right. \\ \left.
		f_{4\bm{\Delta}} \sum_{i\in \mathcal{Z}^c}\sum_{j\in \mathcal{Z}^c\setminus\{i\}} \varepsilon_i\varepsilon_j +
		f_{5\bm{\Delta}} \sum_{i\in \mathcal{Z}}\sum_{j\in \mathcal{Z}^c} \varepsilon_i\varepsilon_j
	\right). \label{eq:ewe1}
\end{multline}
Equation \eqref{eq:sum_ez} implies that 
\begin{equation}
	\sum_{i\in \mathcal{Z}}\sum_{j\in \mathcal{Z}^c} \varepsilon_i\varepsilon_j=\sum_{i\in \mathcal{Z}}\sum_{j\in \mathcal{Z}} \varepsilon_i\varepsilon_j=\sum_{i\in \mathcal{Z}^c}\sum_{j\in \mathcal{Z}^c} \varepsilon_i\varepsilon_j = 0.
\end{equation}
Furthermore, the latter two sums can be broken down as sums of the ``diagonal'' and ``off-diagonal'' residuals which means that
\begin{equation}
	\sum_{i\in \mathcal{Z}} \varepsilon_i^2 = -\sum_{i\in \mathcal{Z}}\sum_{j\in \mathcal{Z}\setminus\{i\}} \varepsilon_i\varepsilon_j, \quad \sum_{i\in \mathcal{Z}^c} \varepsilon_i^2 = - \sum_{i\in \mathcal{Z}^c}\sum_{j\in \mathcal{Z}^c\setminus\{i\}} \varepsilon_i\varepsilon_j.
\end{equation}
Equation \eqref{eq:ewe1} can therefore be written as
\begin{equation}
	E_{\mathcal{W}_{\bm{\Delta}} }\left(\overline{\varepsilon }_{1}^2\right) =
	\frac{1}{n_1^2} \left(
		(f_{1\bm{\Delta}}-f_{3\bm{\Delta}}) \sum_{i\in \mathcal{Z}} \varepsilon_i^2 +
		(f_{2\bm{\Delta}}-f_{4\bm{\Delta}}) \sum_{i\in \mathcal{Z}^c} \varepsilon_i^2
	\right). \label{eq:e12}
\end{equation}
In the equation above, $f_{1\bm{\Delta}}=\Pr(W=1|Z=1,\bm\Delta)$ and $f_{3\bm{\Delta}}=\Pr(W_i=1 \cap W_j=1|\bm\Delta)$ for $i,j \in \mathcal{Z}, i \neq j$. To get an expression for $E_{\mathcal{W}_{\bm{\Delta}} }\left(\overline{\varepsilon }_{0}^2\right)$, we can use the fact that $\Pr(W_i=0 \cap W_j=0|\bm\Delta)=\Pr(W_i=0|\bm\Delta)+\Pr(W_j=0|\bm\Delta)-\Pr(W_i=0 \cup W_j=0|\bm\Delta)=2(1-f_{1\bm{\Delta}})-(1 - f_{3\bm{\Delta}})$, with the corresponding expression for $i,j \in\mathcal{Z}^c$ being $2(1-f_{2\bm{\Delta}})-(1 - f_{4\bm{\Delta}})$. We get
\begin{align}
	E_{\mathcal{W}_{\bm{\Delta}} }\left(\overline{\varepsilon }_{0}^2\right) &=
	\frac{1}{n_0^2} \left(
		(1 - f_{1\bm{\Delta}} - 2(1-f_{1\bm{\Delta}})+(1 - f_{3\bm{\Delta}}) \sum_{i\in \mathcal{Z}} \varepsilon_i^2 +
		(1 - f_{2\bm{\Delta}} - 2(1-f_{2\bm{\Delta}})+(1 - f_{4\bm{\Delta}}) \sum_{i\in \mathcal{Z}^c} \varepsilon_i^2\right) \nonumber \\
		&=
	\frac{1}{n_0^2} \left(
		(f_{1\bm{\Delta}}-f_{3\bm{\Delta}}) \sum_{i\in \mathcal{Z}} \varepsilon_i^2 +
		(f_{2\bm{\Delta}}-f_{4\bm{\Delta}}) \sum_{i\in \mathcal{Z}^c} \varepsilon_i^2
	\right). \label{eq:e02}
\end{align}
Finally, we have
\begin{align}
	E_{\mathcal{W}_{\bm{\Delta}} }\left(\overline{\varepsilon }_{0}\overline{\varepsilon }_{1}\right) &=
	\frac{1}{n_0n_1} \left(
		f_{6\bm{\Delta}}\sum_{i\in \mathcal{Z}}\sum_{j\in \mathcal{Z}\setminus\{i\}} \varepsilon_i\varepsilon_j +
		f_{7\bm{\Delta}}\sum_{i\in \mathcal{Z}^c}\sum_{j\in \mathcal{Z}^c\setminus\{i\}} \varepsilon_i\varepsilon_j +
		(f_{8\bm{\Delta}} + f_{9\bm{\Delta}})\sum_{i\in \mathcal{Z}}\sum_{j\in \mathcal{Z}^c} \varepsilon_i\varepsilon_j
	\right) \nonumber \\
	&=
	\frac{1}{n_0n_1} \left(
		-f_{6\bm{\Delta}}\sum_{i\in \mathcal{Z}} \varepsilon_i^2 -
		f_{7\bm{\Delta}}\sum_{i\in \mathcal{Z}^c} \varepsilon_i^2
	\right). \label{eq:e0e1}
\end{align}
Combining equations \eqref{eq:e12}, \eqref{eq:e02} and \eqref{eq:e0e1}, we get
\begin{multline}
	V_{\mathcal{W}_{\bm{\Delta}} }(\overline{\varepsilon}_1 - \overline{\varepsilon}_0) =
	\left(
		\frac{f_{1\bm{\Delta}} - f_{3\bm{\Delta}}}{n_1^2} +
		\frac{f_{1\bm{\Delta}} - f_{3\bm{\Delta}}}{n_0^2} +
		2\frac{f_{6\bm{\Delta}}}{n_0n_1}
	\right)
	\sum_{i\in \mathcal{Z}} \varepsilon_i^2 + \\
	\left(
		\frac{f_{2\bm{\Delta}} - f_{4\bm{\Delta}}}{n_1^2} +
		\frac{f_{2\bm{\Delta}} - f_{4\bm{\Delta}}}{n_0^2} +
		2\frac{f_{7\bm{\Delta}}}{n_0n_1}
	\right)
	\sum_{i\in \mathcal{Z}^c} \varepsilon_i^2 \label{eq:full_var_epsilon}
\end{multline}

Using $\Pr(Z=1)=n_Z/n$, $\Pr(W=1)=n_1/n$ and $\Pr(Z=1|W=1,\bm\Delta)=\overline{z}_1 = \bm\Delta + \overline{z}_0$ together with Bayes' rule and $\overline{z}_0=(n_Z-n_1\overline{z}_1)/n_0$, we get
\begin{equation}
	f_{1\bm{\Delta}} = \frac{n_1}{n} + \frac{n_0n_1}{nn_Z}\bm\Delta.
\end{equation}
As $f_{1\bm{\Delta}}$ is the probability that a unit with $Z=1$ is treated, $f_{3\bm{\Delta}}$ is the probability that two units with $Z=1$ both are treated. There are $n_Z(n_Z-1)$ combinations of units (twice that of $\binom{n_Z}{2}$ due to symmetry) with $Z=1$ out of which $n_Zf_{1\bm{\Delta}}(n_Zf_{1\bm{\Delta}}-1)$ are the case when both are treated. Therefore the probability of both being treated is
\begin{equation}
	f_{3\bm{\Delta}} = \frac{f_{1\bm{\Delta}}(n_Zf_{1\bm{\Delta}}-1)}{n_Z-1}.
\end{equation}
Similarly, the probability that for two units with $Z=1$, one is in the treatment group, and the other is in the control group is
\begin{equation}
	f_{6\bm{\Delta}} = \frac{n_Zf_{1\bm{\Delta}}(1-f_{1\bm{\Delta}})}{n_Z-1}.
\end{equation}
For the second part of equation \eqref{eq:full_var_epsilon}, it is the case that $f_{2\bm{\Delta}}=\frac{n_1}{n} - \frac{n_0n_1}{n(n-n_Z)}\bm\Delta$, with $f_{4\bm{\Delta}}$ and $f_{7\bm{\Delta}}$ following in the same way as above. Putting it all together and simplifying, we get
\begin{multline}
	V_{\mathcal{W}_{\bm{\Delta}} }(\overline{\varepsilon}_1 - \overline{\varepsilon}_0) =
	\left(
		\frac{n_Z}{n_0n_1(n_Z-1)} +
		\frac{n_0-n_1}{n_0n_1(n_Z-1)}\bm\Delta -
		\frac{1}{n_Z(n_Z-1)}\bm\Delta^2
	\right)
	\sum_{i\in \mathcal{Z}} \varepsilon_i^2 + \\
	\left(
		\frac{n-n_Z}{n_0n_1(n-n_Z-1)} +
		\frac{n_1-n_0}{n_0n_1(n-n_Z-1)}\bm\Delta -
		\frac{1}{(n-n_Z)(n-n_Z-1)}\bm\Delta^2
	\right)
	\sum_{i\in \mathcal{Z}^c} \varepsilon_i^2.
\end{multline}
For the special case with $n_0=n_1$ (balanced experiment), this expression simplifies to
\begin{multline}
	V_{\mathcal{W}_{\bm{\Delta}} }(\overline{\varepsilon}_1 - \overline{\varepsilon}_0) =
	\left(
		\frac{4n_Z}{n^2(n_Z-1)} -
		\frac{1}{n_Z(n_Z-1)}\bm\Delta^2
	\right)
	\sum_{i\in \mathcal{Z}} \varepsilon_i^2 + \\
	\left(
		\frac{4(n-n_Z)}{n^2(n-n_Z-1)} -
		\frac{1}{(n-n_Z)(n-n_Z-1)}\bm\Delta^2
	\right)
	\sum_{i\in \mathcal{Z}^c} \varepsilon_i^2.
\end{multline}
We see that in this case, the variance has its maximum value for $\bm\Delta=0$ and is symmetrically decreasing as the magnitude of $\bm\Delta$ increases.

\end{document}